\documentclass[12pt]{amsart}

\usepackage[pdfauthor   = {Mohamed\ Barakat\ and\ Robin\ Brueser\ and\ Claus\ Fieker\ and\ Tobias\ Huber\ and\ Jan\ Piclum},
            pdftitle    = {Feynman\ integral\ reduction\ using\ Gr\"obner\ bases},
            pdfkeywords = {Loop\ integrals;\ rational\ double\ shift algebra;\ integration-by-parts\ reduction;\ commutative\ and\ noncommutative\ Gr\"obner\ bases;\ computer\ algebra},
            pdfsubject  = {},
            bookmarks=true,
            bookmarksopen=true,
            pagebackref=true,
            hyperindex=true,
            colorlinks=true,
            linkcolor=blue,
            citecolor=blue,
            filecolor=blue,
            urlcolor=blue,
            ]{hyperref}

\usepackage[utf8]{inputenc}
\usepackage[english]{babel}
\usepackage[T1]{fontenc}
\usepackage{geometry}                
\usepackage{cite} 

\usepackage{times}
\usepackage{mathrsfs}
\usepackage{mathtools}
\usepackage{latexsym}
\usepackage{amssymb}
\usepackage{amsthm}
\usepackage{epsfig}
\usepackage{colortbl}
\usepackage[all]{xy}
\usepackage{fancyvrb}
\usepackage{graphicx}
\usepackage[font=small,labelfont=bf]{caption}
\usepackage[dvipsnames]{xcolor}
\usepackage{accents} 
\usepackage{enumerate}

\counterwithin{equation}{section}

\usepackage{wrapfig}
\usepackage{tikz}
\usetikzlibrary{shapes,arrows,matrix,backgrounds,positioning,plotmarks,calc,patterns,
decorations.shapes,
decorations.fractals,
decorations.markings,
decorations.pathreplacing,
decorations.pathmorphing,
decorations.text
}
\usepackage[colorinlistoftodos,shadow]{todonotes}

\usepackage{bm}
\usepackage{multirow}
\usepackage{mdwlist}
\usepackage{stmaryrd}
\usepackage{mathdots} 

\usepackage[toc,page]{appendix}
\usepackage{float}

\usepackage{bigfoot}

\usepackage[sort&compress,capitalise]{cleveref}

\crefname{exmp}{Example}{Examples}

\usepackage[linesnumbered,commentsnumbered,ruled,vlined]{algorithm2e}

\SetCommentSty{mycommfont}

\newtheoremstyle{mytheoremstyle} 
    {5pt}                    
    {5pt}                    
    {\itshape}                   
    {\parindent}                           
    {\bf}                   
    {.}                          
    {.5em}                       
    {}  

\theoremstyle{mytheoremstyle}

\newtheorem{theorem}{Theorem}[section]

\newtheorem{conj}[theorem]{Conjecture}

\newtheorem{prop}[theorem]{Proposition}

\newtheoremstyle{mytdefintionstyle} 
    {5pt}                    
    {5pt}                    
    {\rm}                   
    {\parindent}                           
    {\bf}                   
    {.}                          
    {.5em}                       
    {}  

\theoremstyle{remark}
\newtheorem{rmrk}[theorem]{Remark}

\theoremstyle{mytdefintionstyle}
\newtheorem{defn}[theorem]{Definition}
\newtheorem{exmp}[theorem]{Example}

\newtheoremstyle{exmp_contd}
    {5pt}                    
    {5pt}                    
    {\rm}                   
    {\parindent}                           
    {\bf}                   
    {.}                          
    {.5em}                       
    {\thmname{#1}\ \thmnumber{ #2}\thmnote{#3}\ (continued)}  
\theoremstyle{exmp_contd}

\newcommand{\Q}{\mathbb{Q}}

\newcommand\N{\mathbb{N}}
\renewcommand\phi{\varphi}

\definecolor{darkgray}{rgb}{0.3,0.3,0.3}
\definecolor{LightGray}{gray}{0.9}

\pgfarrowsdeclarecombine[\pgflinewidth]
  {doublestealth}{doublestealth}{stealth'}{stealth'}{stealth'}{stealth'}

\definecolor{darkgreen}{rgb}{0.008,0.617,0.067}
\definecolor{brown}{rgb}{0.6,0.4,0.2}

\newif\ifjournalversion

\author[M.~Barakat]{Mohamed Barakat}
\address{Department of mathematics, Universit\"at Siegen, Walter-Flex-Str.~3, 57068 Siegen, Germany}
\email{\href{mailto:Mohamed Barakat <mohamed.barakat@uni-siegen.de>}{mohamed.barakat@uni-siegen.de}}

\author[R.~Br\"user]{Robin Br\"user}
\address{Theoretische Physik 1, Center for Particle Physics Siegen (CPPS), Universit\"at Siegen, Walter-Flex-Str.~3, 57068 Siegen, Germany
\footnote{Address after 30 September 2022: Albert-Ludwigs-Universit\"at Freiburg, Physikalisches Institut, Hermann-Herder-Str.~3, 79104 Freiburg, Germany}}
\email{\href{mailto:Robin Br\"user <robin.brueser@physik.uni-freiburg.de>}{robin.brueser@physik.uni-freiburg.de}}

\author[C.~Fieker]{Claus Fieker}
\address{TU Kaiserslautern, Gottlieb-Daimler-Straße
Geb\"aude 48, 67663 Kaiserslautern}
\email{\href{mailto:Claus Fieker <fieker@mathematik.uni-kl.de>}{fieker@mathematik.uni-kl.de}}

\author[T.~Huber]{Tobias Huber}
\address{Theoretische Physik 1, Center for Particle Physics Siegen (CPPS), Universit\"at Siegen, Walter-Flex-Str.~3, 57068 Siegen, Germany}
\email{\href{mailto:Tobias Huber <huber@physik.uni-siegen.de>}{huber@physik.uni-siegen.de}}

\author[J.~Piclum]{Jan Piclum}
\address{Theoretische Physik 1, Center for Particle Physics Siegen (CPPS), Universit\"at Siegen, Walter-Flex-Str.~3, 57068 Siegen, Germany}
\email{\href{mailto:Jan Piclum <piclum@physik.uni-siegen.de>}{piclum@physik.uni-siegen.de}}

\begin{document}

\vspace*{-13ex}
\hfill
\begin{minipage}{7em}
\noindent
{\small
\begin{tabular}{r}
SI-HEP-2022-30 \\
P3H-22-101
\end{tabular}}%
\end{minipage}
\vspace{5ex}

\title[Feynman integral reduction using Gr\"obner bases]{Feynman integral reduction using Gr\"obner bases}
\begin{abstract}
We investigate the reduction of Feynman integrals to master integrals using Gr\"obner bases in a rational double-shift algebra $Y$ in which the integration-by-parts (IBP) relations form a left ideal.
The problem of reducing a given family of integrals to master integrals can then be solved once and for all by computing the Gr\"obner basis of the left ideal formed by the IBP relations.
We demonstrate this explicitly for several examples.
We introduce so-called first-order normal-form IBP relations which we obtain by reducing the shift operators in $Y$ modulo the Gr\"obner basis of the left ideal of IBP relations.
For more complicated cases, where the Gr\"obner basis is computationally expensive, we develop an ansatz based on linear algebra over a function field to obtain the normal-form IBP relations.
\end{abstract}



\maketitle

\setcounter{tocdepth}{1}
\tableofcontents

\section{Introduction}

The LHC has been running for more than a decade now and has produced numerous interesting results, among them the discovery of the Higgs boson, precision measurements of Standard Model parameters like the top-quark mass, and searches for physics beyond the Standard Model.
On the theoretical side, all of these studies require precise evaluations of signal and background processes (see for example \cite{HHJP}).
In perturbative quantum field theory this entails the calculation of Feynman integrals with many loops and often with many kinematic invariants like masses and scalar products of external momenta.

A Feynman diagram corresponds in a well-defined way to a loop integral of the form
\begin{align}
  I(z_1, \ldots, z_n)
  =
  \int
  \operatorname{d}^d \ell_1 \cdots \operatorname{d}^d \ell_L
  \frac{1}{P_1^{z_1} \cdots P_n^{z_n}} \mbox{,}
\end{align}
where $d$  is the space-time dimension in dimensional regularization and $\ell_i$ with $i = 1, \ldots, L$ are the \textbf{loop momenta}.
The propagators (more precisely, propagator denominators) $P_i$ with $i = 1, \ldots, n$ are usually of the form $m_i^2-p_i^2$, where $m_i$ is a particle mass and $p_i$ is a linear combination of the $L$ loop momenta and $E$ \textbf{external momenta} $k_1, \ldots, k_E$.
They are deduced from the process under consideration and, without loss of generality, can be assumed to be linearly independent.
We call the set $\{ I(z_1, \ldots, z_n) \mid z_i\in\mathbb{Z} \}$ with given propagators $P_i$ a \emph{family of integrals}.
The integers $z_i$ are called the \textbf{indices} of the integral.
A given physical process is usually expressed in terms of integrals of several different families. 
Note that we have suppressed the integral's dependence on the kinematic invariants, since we are here mostly concerned with the dependence on the indices.
For a modern account on the calculation of loop integrals see for example \cite{Weinzierl:2022eaz}.

In a typical calculation, one has to evaluate thousands of loop integrals belonging to several integral families.
An indispensable tool in the calculation of multiloop integrals is therefore the integration-by-parts (IBP) method~\cite{Tkachov:1981wb,Chetyrkin:1981qh}, which provides relations between loop integrals with different propagator and numerator powers.
These recurrence relations are called IBP relations and can be used to express all integrals of a given family in terms of a small number of so-called \emph{master integrals}.
Nowadays this is usually applied in the form of Laporta's algorithm~\cite{Laporta:2000dsw}, which solves the system of linear equations generated by plugging in numerical values for the indices $z_1, \ldots, z_n$.

The performance of Laporta's algorithm can be greatly improved by using modular arithmetic~\cite{vonManteuffel:2014ixa,Peraro:2016wsq}.
This avoids huge intermediate expressions during the calculation and allows for more efficient parallelization.
There are many public and private codes to perform the integration by parts reduction, for instance ${\mathtt{AIR}}$~\cite{Anastasiou:2004vj}, ${\mathtt{FIRE}}$~\cite{Smirnov:2008iw,Smirnov:2014hma,Smirnov:2019qkx},
${\mathtt{Reduze}}$~\cite{Studerus:2009ye,vonManteuffel:2012np}, ${\mathtt{Kira}}$~\cite{Maierhofer:2017gsa,Klappert:2020nbg}, and \textsc{FiniteFlow} \cite{Peraro_2019}.

In order to generate a linear system of equations for Laporta's algorithm one has to specialize a set of IBP relations to a range of indices $z_1, \ldots, z_n$.
Typically this linear system contains a large number of integrals (as unknowns) that are oftentimes not directly needed but must be included to ensure the full reduction of the desired integrals.
This problem can be (partially) avoided by starting with a set of so-called unitarity-compatible IBP relations based on syzygies \cite{Bern:1993kr,Gluza:2010ws,Schabinger:2011dz,Lee:2014tja,Bohm:2017qme,Kosower:2018obg} (over a polynomial ring).
We refer to these IBP relations in \Cref{sec:special_ibps} as \emph{special} IBP relations.
They reduce the size of the linear system and improve the performance.
Additional new ideas towards a more direct reduction procedure have been developed:
They rely on algebraic geometry~\cite{Larsen:2015ped,Bohm:2018bdy,Bendle:2019csk} and intersection theory~\cite{Mastrolia:2018uzb,Frellesvig:2019uqt,Frellesvig:2019kgj,Abreu:2019wzk,Frellesvig:2020qot,Weinzierl:2020xyy,Caron-Huot:2021iev,Chen:2022lzr,Chestnov:2022alh}.

One limitation of Laporta's algorithm is that the reduction is only found for a given list of integrals.
Thus, when additional integrals are needed at a later point, the program has to be run again.
This can be overcome by deriving a full solution to the system of IBP recurrence relations.
Since a parametric solution by hand is clearly not feasible for multiscale problems, it would be desirable to have an algorithmic way of solving the IBP relations once and for all.
${\mathtt{LiteRed}}$~\cite{Lee:2012cn} is a publicly available program that performs this task using tailored heuristics that is able to reduce several complicated integral families.
In our work, we investigate the application of Gr\"obner bases to the solution of IBP recurrence relations.

Previous application of Gr\"obner bases in this context can be found in \cite{Tarasov:2004ks,Gerdt:2005qf,Smirnov:2005ky,Smirnov:2006tz,Smirnov:2006wh,Lee:2008tj}.
In \cite{Tarasov:2004ks}, the IBP relations are first transformed into a system of partial differential equations for which the Gr\"obner basis is then computed.
This method requires that all propagators have different, non-zero masses and that no external momentum squared is equal to one of the masses squared.
Thus, it is not always possible to apply the result to cases with zero or equal masses or on-shell momenta, since such limits can be singular.
Reference~\cite{Smirnov:2005ky} uses a modified version of Buchberger's algorithm to obtain so-called sector bases~\cite{Smirnov:2006tz}.
Finally we would like to emphasize that -- in contrast to previous noncommutative Gr\"obner basis approaches -- we work in the \emph{rational} double-shift algebra defined in \Cref{sec:Y}.

This paper is organized as follows.
In \Cref{sec:Y} we define the (rational) double-shift algebra in which the IBP relations form a left ideal.
In \Cref{sec:Groebner} we introduce the first-order normal-form IBP relations and highlight in which sense they differ from other well-known sets of IBP relations.
Furthermore we introduce the notion of a first-order family, i.e., an integral family for which the left ideal of IBP relations is generated by the first-order normal-form IBP relations.
In \Cref{sec:sectors} we define the notion of formally scaleless monomials and relate them to scaleless sectors.
In \Cref{sec:examples} we demonstrate on selected examples of first-order families the computation of Gr\"obner bases, normal-form IBP relations, and the detection of scaleless sectors using Gr\"obner basis reductions.
\Cref{sec:special_ibps} recalls the construction of special IBP relations using syzygies in a polynomial ring.
The special IBP relations turn out to provide a more efficient set of generators as a starting point for the Linear Algebra Ansatz in \Cref{sec:LA-Ansatz} to compute the normal-form IBP relations without precomputing a Gr\"obner basis.
Finally we conclude in \Cref{sec:conclusion}.

\section{The left ideal of IBP relations in the rational double-shift algebra} \label{sec:Y}

\subsection{Notation}

Denote by $p_i$ (for a symbol $p$ and a natural number $i$) the column vector of $d'$ indeterminates
\begin{align}
  p_i =
  \begin{pmatrix}
    p_i^0 \\
    p_i^1 \\
    \vdots \\
    p_i^{d'-1}
  \end{pmatrix}
\end{align}
and define the Lorentz invariant quadratic expression
\begin{align}
  p_i \cdot p_j = p_i^0 p_j^0 - \sum_{\mu=1}^{d'-1} p_i^\mu p_j^\mu \mbox{.}
\end{align}
For $p = \ell$ the vectors $\ell_1, \ldots, \ell_L$ will refer to the $L$ \textbf{loop momenta}.
For $p = k$ the vectors $k_1, \ldots, k_E$ will refer to the $E$ \textbf{external momenta}.

Consider the polynomial algebra $\Q[d, m_i^2]$ with coefficients in the field $\Q$ of rational numbers.
Its elements are polynomial expressions in the dimension symbol $d$ and the symbols of squared masses $m_i^2$.
Define the field of rational functions
\begin{align}
  \mathbb{F} \coloneqq \Q(d, m_i^2) \coloneqq \left\{ \frac{Z}{N} \,\Big|\, Z, N \in \Q[d, m_i^2], N \neq 0 \right\}
\end{align}
where the numerators $Z$ and nonzero denominators $N$ are polynomials in $\Q[d, m_i^2]$.

Consider the Lorentz invariant expressions that are polynomial expressions in the scalar products of the $L+E$ momenta $\ell_1, \ldots, \ell_L, k_1, \ldots, k_E$ with coefficients in $\mathbb{F}$.
Each such expression can be written as a polynomial in the $n = {\frac{L(L+1)}{2}+LE}$ \textbf{propagators} $
P_1, \ldots, P_n$ and so-called \textbf{extra Lorentz invariants} $S_1, \ldots, S_q$ with coefficients in $\mathbb{F}$.
The extra Lorentz invariants are constructed from the external momenta in such a way that the set $\{ P_1, \ldots, P_n, S_1, \ldots, S_q \}$ is algebraically independent over $\mathbb{F}$.
This means that $P_1, \ldots, P_n, S_1, \ldots, S_q$ generate a polynomial algebra over $\mathbb{F}$, which we denote by
\begin{align}
  T = \mathbb{F}[S_1, \ldots, S_q][P_1, \ldots, P_n] \mbox{.}
\end{align}
Since from some point on we do not need the special form of the $P_i$'s and $S_j$'s we replace them by symbols $D_1, \ldots, D_n$ and $s_1, \ldots, s_q$, respectively.
Likewise we replace the polynomial algebra $T$ by the isomorphic polynomial algebra
\begin{align}
  R = \mathbb{F}[s_1, \ldots, s_q][D_1, \ldots, D_n] \mbox{.}
\end{align}
For more mathematical details on the construction of the polynomial algebras $T$ and $R$ see \Cref{sec:App}.

The IBP relations are obtained from the fact that the operator $\frac{\partial}{\partial \ell_i^\mu} v_i^\mu$ turns the loop integrand into a divergence, i.e., annihilates the loop integral in dimensional regularization.
More precisely:
\begin{align}
  0
  =&
  \int
  \operatorname{d}^d \ell_1 \cdots \operatorname{d}^d \ell_L
  \; \frac{\partial}{\partial \ell_i^\mu} \left(v_i^\mu \;
  \frac{1}{P_1^{z_1} \cdots P_n^{z_n}} \right)
  \label{eq:div_integrand}
  \\
  =&
  \int
  \operatorname{d}^d \ell_1 \cdots \operatorname{d}^d \ell_L
  \left(\frac{\partial v_i^\mu}{\partial \ell_i^\mu} \right)
  \frac{1}{P_1^{z_1} \cdots P_n^{z_n}} + \label{eq:div}
  \\
  &
  \int
  \operatorname{d}^d \ell_1 \cdots \operatorname{d}^d \ell_L
  \left(v_i^\mu \frac{\partial}{\partial \ell_i^\mu} \right)
  \frac{1}{P_1^{z_1} \cdots P_n^{z_n}} \mbox{,} \label{eq:diff}
\end{align}
where
  \begin{align}
    v_i = C^j_i B_j
  \end{align}
  for $B_j \in \{\ell_1,\ldots,\ell_L, k_1, \ldots, k_E\}$ with coefficients (column) vector
\begin{align} \label{eq:C}
  C = (C^j_i)_{i=1,\ldots, L, j = 1, \ldots, L+E} \in T^{L(L+E) \times 1} \mbox{.}
\end{align}
The standard IBP relations are obtained by $C$ running through the standard basis of $T^{L(L+E) \times 1}$ (see \eqref{eq:standard_IBPs} below).

In the following we will rewrite the expression $\left(\frac{\partial v_i^\mu}{\partial \ell_i^\mu}\right)$ in \eqref{eq:div} and the differential operator $\left(v_i^\mu \frac{\partial}{\partial \ell_i^\mu}\right)$ in \eqref{eq:diff} in terms of the ring $R$.
To this end we define the \textbf{IBP-generating matrix}\footnote{The name is motivated by equation \eqref{eq:IBP}.} as the product matrix
\begin{align} \label{eq:Euler}
  \mathcal{E}
  &\phantom{:}=
  \left(\mathcal{E}_{j,c}^i\right) \\
  &\coloneqq
  J \cdot \underbrace{\left[I_L \otimes \begin{pmatrix} \ell_1 & \cdots & \ell_L & k_1 & \cdots & k_E \end{pmatrix}\right]}_{\in T^{Ld' \times L(L+E)}} = \left( \frac{\partial P_c}{\partial \ell_i^\mu} B^\mu_j \right) \in T^{n \times L(L+E)} \subset \widetilde{T}^{n \times L(L+E)} \mbox{,} \nonumber
\end{align}
where $J \coloneqq \begin{pmatrix} \frac{\partial P_c}{\partial \ell_i^\mu} \end{pmatrix} \in \widetilde{T}^{n \times Ld'}$ is the Jacobian matrix of the propagators, and where $\widetilde{T}$ is defined in \Cref{sec:App}.
Like the propagators, and unlike the Jacobian matrix, the entries of the IBP-generating matrix belong to the subring $T$ and can therefore be effectively rewritten as matrices over $R \cong T$ using the subalgebra membership algorithm.
The latter can be replaced by simple linear algebra due to the affine nature of $P_i$ as expressions in $p_i \cdot p_j$.
The dimensions of $\mathcal{E}$ are already independent of $d'$.
However, its entries as expressions in the generators of the subring $T$ formally still depend on $d'$.
But once $\mathcal{E}$ is rewritten as a matrix over $R$, the initial dependency of $\mathcal{E} \in R^{n \times L(L+E)}$ on the dimension $d'$ disappears\footnote{Physically, $d'$ should be thought of as the symbolic regularizing dimension $d$ rather than an integer.}.

For the coefficients vector $C \in R^{L(L+E) \times 1}$ consider the Jacobian
\begin{align}
  J_{C} \coloneqq \left( \frac{\partial C^j_i}{\partial D_c} \right) \in R^{L(L+E) \times n}
\end{align}
and the square matrix
\begin{align}
  \mathcal{E}_C \coloneqq \mathcal{E} J_C \in R^{n \times n} \mbox{.}
\end{align}
The divergence summand in \eqref{eq:div} becomes
\begin{align} \tag{\ref{eq:div}'}
  \frac{\partial v^\mu_i}{\partial \ell^\mu_i}= d \cdot C^i_i + \operatorname{tr} \mathcal{E}_C \coloneqq d \cdot \sum_{i=1}^L C^i_i + \operatorname{tr} \mathcal{E}_C \in R \mbox{.}
\end{align}
Furthermore, the second summand \eqref{eq:diff} becomes
\begin{align} \tag{\ref{eq:diff}'}
  v_i^\mu \frac{\partial}{\partial \ell_i^\mu}
  &=
  C^j_i \left( B_j^\mu \frac{\partial D_b}{\partial \ell_i^\mu} \right) \frac{\partial}{\partial D_b}
  =
  \mathcal{E}_{j,b}^i C^j_i \frac{\partial}{\partial D_b}
  \coloneqq
  \sum_{b=1}^n \left(\sum_{j=1}^{L+E} \sum_{i=1}^L \mathcal{E}_{j,b}^i C^j_i\right) \frac{\partial}{\partial D_b}\mbox{.}
\end{align}
To determine the action of the differential operation $v_i^\mu \frac{\partial}{\partial \ell_i^\mu}$ on $D_c^{-z_c}$ we use $\frac{\partial}{\partial D_b} D_c^{-z_c} = -z_c \delta_c^b D_c^{-(z_c+1)} = (-z_c D_c^{-1}) \delta_c^b D_c^{-z_c}$ resulting in
\begin{align} \label{eq:shift}
  \left(v_i^\mu \frac{\partial}{\partial \ell_i^\mu}\right) D_c^{-z_c}
  =
  \left(-z_c D_c^{-1} \mathcal{E}_{j,c}^i C^j_i \right) D_c^{-z_c}
  \coloneqq
  -z_c D_c^{-1} \left(\sum_{j=1}^{L+E} \sum_{i=1}^L \mathcal{E}_{j,c}^i C^j_i\right) D_c^{-z_c} \mbox{.}
\end{align}
The next section introduces the shift algebra which contains the IBP relations as shift operators.

\subsection{The (rational) double-shift algebra}

The IBP relations can be understood as shift operators acting on the polynomial algebra
\begin{align} \label{eq:A}
  A \coloneqq \mathbb{F}[s_1,\ldots,s_q][a_1,\ldots,a_n]
\end{align}
by shifts, and therefore as elements of the double-shift algebra
\begin{align}
  Y^\mathrm{pol} \coloneqq A\langle D_1, D_1^-, \ldots, D_n, D_n^- \rangle
\end{align}
with the relations (no summation over repeated indices)
\begin{equation}
\begin{aligned}
  [a_i, D_j] = \delta_{ij} D_i \, , \qquad [a_i, D_j^-] = -\delta_{ij}  D_i^-, \qquad D_i D_i^- = 1,
  \\
  [a_i, a_j] = [D_i, D_j] = [D_i^-, D_j^-] = [D_i, D_j^-] = 0,
\end{aligned}
\end{equation}
and partial right action
\begin{align}
  I(\ldots,z_i,\ldots) &\bullet D_i = I(\ldots,z_i-1,\ldots), \\
  \underbrace{I(\ldots,z_i,\ldots)}_{\text{not scaleless}} &\bullet D_i^- = I(\ldots,z_i+1,\ldots),
  \label{eq:pra1}
  \\
  I(\ldots,z_i,\ldots) &\bullet a_i = z_i I(\ldots,z_i,\ldots) \mbox{.}
  \label{eq:pra2}
\end{align}
The prefix ``double'' refers to the simultaneous occurrence of both the lowering operators $D_i$ and the raising operators $D_i^-$.

The action is \emph{partial} since $D_i^-$ cannot be applied to a scaleless integral.\footnote{Alternatively, one could rephrase such partial actions of algebras as actions of associated algebroids.}
Our choice of the \emph{right} action will be justified in \Cref{rmrk:right_action} and the definition of scaleless integrals is deferred to \Cref{sec:sectors}.

One can extend the action to the rational function field
\begin{align}
  K \coloneqq \operatorname{Frac} A = \mathbb{F}(s_1,\ldots,s_q)(a_1,\ldots,a_n) \coloneqq \left\{ \frac{Z}{N} \,\Big|\, Z, N \in A, N \neq 0 \right\} \mbox{,}
\end{align}
yielding the \textbf{rational} double-shift algebra
\begin{align} \label{eq:Y}
  Y \coloneqq K\langle D_1, D_1^-, \ldots, D_n, D_n^- \rangle \mbox{.}
\end{align}

The partial right action is extended via
\begin{align}
  I(z_1, \ldots, z_n) \bullet \frac{Z(a_1, \ldots, a_n)}{N(a_1,\ldots, a_n)} = \frac{Z(z_1, \ldots, z_n)}{N(z_1,\ldots, z_n)} I(z_1, \ldots, z_n) \mbox{,}
\end{align}
where $Z(a_1, \ldots, a_n)$ and $N(a_1,\ldots, a_n) \neq 0$ are polynomial expressions in the $a_i$'s with coefficients in rational expressions of the kinematic invariants $\mathbb{F}(s_1,\ldots,s_q)$ whenever $N(z_1,\ldots, z_n)$ is nonzero.

\subsection{Generating the left ideal of IBP relations}

For an arbitrary coefficients vector $C \in R^{L(L+E) \times 1}$ we get, using the above summation convention, the IBP (shift) operator
\begin{align} \label{eq:IBP}
  r(C) \coloneqq d \cdot C^i_i + \operatorname{tr} \mathcal{E}_C -a_c D_c^- \mathcal{E}_{j,c}^i C^j_i \in Y^\mathrm{pol} \subset Y \mbox{.}
\end{align}
Note that due to the Jacobian expression entering $\mathcal{E}_C$ the map
\begin{align} \label{eq:r}
  r: R^{L (L+E) \times 1} \to Y^\mathrm{pol}, C \mapsto r(C)
\end{align} is not $R$-linear, but merely linear over the subalgebra $\mathbb{F}[s_1, \ldots s_q] < R$.

The $L(L+E)$ \textbf{standard IBP relations} are obtained by $C$ running through the standard basis $\{e_1, \ldots, e_{L(L+E)}\}$ of $R^{L(L+E) \times 1}$, i.e.,
\begin{align} \label{eq:standard_IBPs}
  r_i \coloneqq r(e_i) \mbox{.}
\end{align}

\begin{defn} \label{defn:I_IBP}
Define the \textbf{left ideal of IBP relations} in $Y^\mathrm{pol}$ and $Y$ as the left ideal generated by IBP relations
\begin{align}
\begin{array}{rll}
  I_\mathrm{IBP}^\mathrm{pol}
  &\coloneqq \langle \operatorname{im}(r) \rangle_{Y^\mathrm{pol}} &= \left\langle r(C) \mid C \in R^{L(L+E) \times 1} \right\rangle_{Y^\mathrm{pol}} \lhd Y^\mathrm{pol},
  \\[0.3em]
  I_\mathrm{IBP}
  &\coloneqq \langle \operatorname{im}(r) \rangle_Y &= \left\langle r(C) \mid C \in R^{L(L+E) \times 1} \right\rangle_Y \quad \lhd Y \mbox{.}
\end{array}
\end{align}
The left ideal of IBP relations annihilates all loop integrals of the given family, i.e.,
  \begin{align}
    I(z_1, \ldots, z_n) \bullet f = 0
  \end{align}
  for all $f \in I_\mathrm{IBP}^{(\mathrm{pol})}$ whenever the partial action is defined.
\end{defn}

  Since the map $r$ in \eqref{eq:r} is not $R$-linear one needs to formally prove that the left ideals $I^\mathrm{pol}_\mathrm{IBP}$ and $I_\mathrm{IBP}$ are \emph{finitely} generated, more precisely:
\begin{prop} \label{prop:finite_generation}
 The left ideals $I^\mathrm{pol}_\mathrm{IBP}$ and $I_\mathrm{IBP}$ are generated by the standard IBP relations:
  \begin{equation}
  \begin{aligned}
    I_\mathrm{IBP}^\mathrm{pol}
    &= \langle r_i \mid i = 1, \ldots, L(L+E) \rangle_{Y^\mathrm{pol}} \lhd Y^\mathrm{pol},
    \\
    I_\mathrm{IBP}
    &= \langle r_i \mid i = 1, \ldots, L(L+E) \rangle_Y \quad \lhd Y \mbox{.}
  \end{aligned}
  \end{equation}
\end{prop}
\begin{proof}
  A vector $C \in T^{L(L+E) \times 1} \equiv R^{L(L+E) \times 1}$ is the common coefficients vector of $v_i = C^j_i B_j$ for $i=1,\ldots,L$ in \eqref{eq:div_integrand}.
  Each coefficient $C_i^j \in T \equiv R$ is an $\mathbb{F}[S_1, \ldots, S_q]$-linear combination of monomials of the form $\prod_{c=1}^n P_c^{z_{c,i,j}}$.
  It follows that
  \begin{align} \label{eq:trick}
    \left( \prod_{c=1}^n P_c^{z_{c,i,j}} \right) B_j \frac{1}{P_1^{z_1} \cdots P_n^{z_n}} = B_j \frac{1}{P_1^{z_1-z_{1,i,j}} \cdots P_n^{z_n-z_{n,i,j}}} \mbox{.}
  \end{align}
  The right hand side of \eqref{eq:trick} results in the IBP operator
  \begin{equation}
    \left( \prod_{c=1}^n D_c^{z_{c,i,j}} \right) r_j \in Y^\mathrm{pol} \mbox{.} \qedhere
  \end{equation}
\end{proof}

Dictated by the loop diagram, some of the propagators play the role of numerators, i.e., only their nonpositive exponents $z_i$ are considered.
These $u$ many propagators are called \textbf{irreducible numerators} and are conventionally grouped at the end:
\begin{align}
  \underbrace{P_1, \ldots, P_{n-u}}_{n-u}, \underbrace{P_{n-u+1}, \ldots, P_n}_{u} \mbox{.}
\end{align}

\begin{rmrk} \label{rmrk:right_action}
  The annihilator of \emph{all} elements of a right or left action is a two-sided ideal.
  However, annihilators of \emph{partial} right (left) actions are merely left (right) ideals.
  But since the software we use for computing noncommutative Gr\"obner bases only supports \emph{left} ideals, we had to opt for partial \emph{right} actions.
\end{rmrk}

\section{Gr\"obner bases in the noncommutative double-shift algebras} \label{sec:Groebner}

\subsection{Gr\"obner bases, standard monomials, and master integrals}

Below we need the notion of a Gr\"obner basis of the left ideals $I_\mathrm{IBP}^\mathrm{pol}$ and $I_\mathrm{IBP}$ in the respective noncommutative algebras $Y^\mathrm{pol} = A\langle D_1, D_1^-, \ldots, D_n, D_n^- \rangle$ and $Y = K\langle D_1, D_1^-, \ldots, D_n, D_n^- \rangle$.
In both cases we use generalizations of Buchberger's algorithm \cite{Buch} to the context of GR-algebras and Ore algebras, respectively.

Replacing a set of generators of a left ideal by a Gr\"obner basis (with respect to a monomial order) might introduce redundant generators.
However, these are necessary for the reduction procedure to produce unique remainders, independent of possible choices of the reduction steps.

Let $G^\mathrm{pol}$ and $G$ denote the Gr\"obner bases of the left ideals $I_\mathrm{IBP}^\mathrm{pol} \unlhd Y^\mathrm{pol}$ and $I_\mathrm{IBP} \unlhd Y$, respectively.
As customary we denote by $\operatorname{NF}_{G^\mathrm{pol}}(f) \in Y^\mathrm{pol}$ and $\operatorname{NF}_G(g) \in Y$ the normal forms of $f \in Y^\mathrm{pol}$ and $g \in Y$ with respect to the Gr\"obner bases $G^\mathrm{pol}$ and $G$, respectively.

One would generally expect the number of elements in $G$ to be smaller or equal to that of $G^\mathrm{pol}$.
This might fail in trivial cases as in \Cref{ex:one_loop_tadpole}.
More involved cases like \Cref{exmp:one-loopbox} show that the difference of cardinalities can be significant.

\begin{defn}
A \textbf{standard monomial} with respect to the Gr\"obner basis $G$ of $I_\mathrm{IBP} \unlhd Y$ is a monomial $f$ in the indeterminates $D_i, D_j^-$ such that $\operatorname{NF}_G(f) = f$.
\end{defn}

\begin{rmrk} \label{rmrk:standard_monomials}
The set of standard monomials is a basis for the finite dimensional $K$-vector space $Y / I_\mathrm{IBP}$.
The set of standard monomials corresponds to a set of master integrals with respect to some fixed initial integral, usually $I(\underbrace{1, \ldots, 1}_{n-u}, \underbrace{0, \ldots, 0}_u)$.
Note that due to possible symmetries of the problem there might exist $\mathbb{F}(s_1,\ldots,s_q)$-linear relations among these master integrals (cf.~\Cref{exmp:2LoopTadpole}).
\end{rmrk}

For Gr\"obner basis and normal form computations in the polynomial double-shift algebra $Y^\mathrm{pol}$ we use \textsc{Singular}'s subsystem \textsc{Plural} \cite{plural} and for the rational double-shift algebra $Y$ we use Chyzak's \textsf{Maple} package $\mathtt{Ore\_algebra}$ \cite{Chyzak-1998-GBS}.
For the technical implementation we developed the package $\mathtt{LoopIntegrals}$ \cite{LoopIntegrals}.
$\mathtt{LoopIntegrals}$ is currently written in {\textsf{GAP}}~\cite{GAP4111} and relies on the $\mathtt{homalg}$-project packages~\cite{homalg-project} which offer a unified interface to \textsc{Singular} \cite{singular431} and \textsf{Maple}.
A \textsc{Mathematica} package that can perform the required Gröbner basis computations over the rational double-shift algebra is $\mathtt{HolonomicFunctions}$ \cite{Koutschan09,Koutschan10b}.
The $\mathtt{homalg}$-project does not offer an interface to \textsc{Mathematica} yet.

\begin{exmp}[One-loop tadpole] \label{ex:one_loop_tadpole}
  The one-loop tadpole is defined by the loop momentum $\ell_1$ and no external momentum (in particular, $L=1, E=0$).
  The single internal line is massive with mass $m$.
  The $n=1$ propagator is
  \begin{align}
    P_1 = -\left(\ell_1^2 - m^2\right) \mbox{.}
  \end{align}
  The $L(L+E) = 1$ standard IBP relation is
  \begin{align}
    r_1 = 2 m^2 a_1 D_1^-+ (d - 2 a_1) \mbox{,}
  \end{align}
  expressed as element of the polynomial double-shift algebra
  \begin{align}
    Y^\mathrm{pol} \coloneqq \Q(d,m^2)[a_1] \langle D_1, D_1^- \rangle \mbox{.}
  \end{align}
  
  One can verify that the cyclic generator $r_1$ is already the reduced Gr\"obner basis $G^\mathrm{pol} = \{ r_1 \}$ of the left ideal $I_\mathrm{IBP}^\mathrm{pol} \unlhd Y^\mathrm{pol}$.
  Switching to the \emph{rational} double-shift algebra $Y$ and computing the reduced Gr\"obner basis $G$ of $I_\mathrm{IBP} \unlhd Y$ we get the two generators
  \begin{align} 
    G = \{ G_1, G_2 \} = 
    \{ 2 m^2 a_1 D_1^- + (d - 2 a_1),
       2 m^2 (a_1 - 1) + (d - 2 a_1 + 2) D_1 \} \mbox{,}
  \end{align}
  where $G_2 = D_1 G_1$.
  A simple computation reveals that the Gr\"obner basis reductions modulo $G^\mathrm{pol} \subset Y^\mathrm{pol}$ and $G \subset Y$ yield
  \begin{align}
    \operatorname{NF}_G(a_1 D_1^-) = \operatorname{NF}_{G^\mathrm{pol}}(a_1 D_1^-) = -\frac{d - 2 a_1}{2m^2} \mbox{.}
  \end{align}
  In particular, the normalized IBP relation $\frac{r_1}{2m^2}$ takes the special form:
  \begin{align}
    \frac{r_1}{2m^2} = a_1 D_1^- + \frac{d - 2 a_1}{2m^2} = a_1 D_1^- - \operatorname{NF}_G(a_1 D_1^-) \mbox{.}
  \end{align}
  
  This special form of the IBP relation also reflects the functional dependence of the integral's closed-form result:
  Using
  \begin{align}
    I(z_1) = \int \operatorname{d}^d \ell_1 \frac{1}{\left(-\ell_1^2 + m^2 \right)^{z_1}}
    =
    i\pi^{d/2}\,
    \frac{
    \Gamma\left(z_1 - \frac{d}{2}\right)}{
    \Gamma(z_1)\,
    (m^2)^{z_1 - d/2}}
  \end{align}
  and $\Gamma(z+1)=z\,\Gamma(z)$ we obtain
  \begin{align} \label{eq:contiguous_one-loop-tadpole}
    z_1\, I(z_1 + 1) =-\frac{d - 2 z_1}{2 m^2}\, I(z_1) \mbox{.}
  \end{align}
  It is obvious from \eqref{eq:contiguous_one-loop-tadpole} that the special form of the IBP relation ---regardless of its characterization using Gr\"obner bases--- is ideally suited for performing IBP reductions.
  
  Finally,
  \begin{equation}
  \begin{aligned}
    \operatorname{NF}_G(1) =\,
    &
    1
    &&\text{$\leadsto 1$ is (trivially) a standard monomial},
    \\
    \operatorname{NF}_G(D_1) =\,
    &
    -\frac{2 m^2 (a_1 - 1)}{d - 2 a_1 + 2}
    &&\text{$\leadsto D_1$ is a nonstandard monomial.}
  \end{aligned}
  \end{equation}
  Hence, the set of standard monomials with respect to $G$ is merely
  \begin{align}
    \{ 1 \} \mbox{,}
  \end{align}
  which by \Cref{rmrk:standard_monomials} corresponds to the single master integral
  \begin{align}
    \{ I(1) \} \mbox{.}
  \end{align}
  
\end{exmp}

\subsection{First-order normal-form IBP relations}

Motivated by the previous simple example we are particularly interested in IBP relations of the following special form.
\begin{defn} \label{defn:1st-order_normal-form}
We call for $i = 1, \ldots, n = {\frac{L(L+1)}{2}+LE} \leq L(L+E) = n + \binom{L}{2}$ the IBP relations of the form
\begin{align} \label{eq:Ri}
  R_i \coloneqq
  \begin{cases}
    a_i D_i^- - \operatorname{NF}_G(a_i D_i^-) & \mbox{ for } i \leq n-u,
    \\[0.2em]
    a_i D_i - \operatorname{NF}_G(a_i D_i) & \mbox{ for } i > n-u \mbox{,}
  \end{cases}
\end{align}
in $I_\mathrm{IBP} \unlhd Y = K\langle D_1, D_1^-, \ldots, D_n, D_n^- \rangle$ where $\operatorname{NF}_G(a_i D_i^{\text{(}-\text{)}})$ is a $K$-linear combination of the standard monomials, the \textbf{first-order normal-form IBPs}.
\end{defn}
The adjective ``first-order'' refers to the linear occurrences of $D_i^-$ and $D_i$ in the monomials in \eqref{eq:Ri} for which the normal forms are to be computed.

Examples (with $u=0$) show that $\operatorname{NF}_{G^\mathrm{pol}}(a_i D_i^-)$ with respect to the \emph{polynomial} Gr\"obner basis $G^\mathrm{pol}$ includes expressions in the $D_j^-$'s.
In contrast, in the examples treated in \Cref{sec:examples}, the normal forms $\operatorname{NF}_G(a_i D_i^-)$ with respect to the \emph{rational} Gr\"obner basis $G$ do not involve any $D_j^-$.
However, the $K$-coefficients of $\operatorname{NF}_G(a_i D_i^-)$ will often be true fractions in $K \setminus A = (\operatorname{Frac} A) \setminus A$.

\begin{rmrk} \label{rmrk:aposteriori}
  We expect the numerators $\operatorname{num}(R_i) \in Y^\mathrm{pol}$ to lie in $I_\mathrm{IBP}^\mathrm{pol}$.
  This can be verified for each specific example by checking that all $\operatorname{num}(R_i)$ reduce to zero modulo the polynomial Gr\"obner basis $G^\mathrm{pol} \subset Y^\mathrm{pol}$.
  This proves a posteriori that the reduction modulo $G$ used to compute the $R_i$'s can be obtained without dividing by polynomials in $\Q[a_1,\ldots,a_n]$.
  In particular, the normal-form IBP relations are valid for all indices $(z_1, \ldots, z_n) \in \mathbb{Z}^n$.
\end{rmrk}

\begin{defn} \label{defn:polynomiality}
  The reductions of $\operatorname{num}(R_i)$ modulo $G^\mathrm{pol}$ to zero for $i=1,\ldots, n$ yield a matrix $\tau \in (Y^\mathrm{pol})^{n \times L(L+E)}$ such that
  \begin{align} \label{eq:tau}
    \begin{pmatrix}
      \operatorname{num}(R_1) \\
      \vdots \\
      \operatorname{num}(R_n)
    \end{pmatrix}
    =
    \tau
    \begin{pmatrix}
      r_1 \\
      \vdots \\
      \vdots \\
      r_{L(L+E)}
    \end{pmatrix} \mbox{.}
  \end{align}
  We call the matrix $\tau$ \textbf{a certificate of polynomiality}.
\end{defn}

\begin{rmrk} \label{rmrk:tau_not_unique}
  Since $\{r_1, \ldots, r_{L(L+E)}\}$ is generally not a \emph{free} subset\footnote{unless when $L(L+E) = 1$} of $Y^\mathrm{pol}$ the equation \eqref{eq:tau} does not determine $\tau$ uniquely.
  The non-freeness of $\{r_1, \ldots, r_{L(L+E)}\}$ (and hence the non-uniqueness of $\tau$) can be verified by computing any nontrivial syzygy of $\left(\begin{smallmatrix} r_1 \\ \vdots \\ r_{L(L+E)} \end{smallmatrix}\right)$ over $Y^\mathrm{pol}$.
\end{rmrk}

\subsection{First-order family of loop integrals}

Examples in \Cref{sec:examples} show that the left ideal $\langle \operatorname{num}(R_i) \mid i = 1, \ldots, n \rangle 
\unlhd Y^\mathrm{pol}$ generated by the numerators $\operatorname{num}(R_i)$ is generally \emph{strictly} contained in $I_\mathrm{IBP}^\mathrm{pol}$.
However, there is a class of loop integrals, for which equality holds over the rational double-shift algebra $Y$:

\begin{defn} \label{defn:1storder}
  We call an integral family \textbf{first-order} if the first-order normal-form IBPs form a generating set of the left ideal $I_\mathrm{IBP} \unlhd Y$ of IBP relations:
  \begin{align}
    \langle R_i \mid i = 1, \ldots, n \rangle_Y = I_\mathrm{IBP} \coloneqq \langle r_i \mid i = 1, \ldots, L(L+E) \rangle_Y \mbox{.}
  \end{align}
\end{defn}

To verify this it suffices to compute the (normalized) reduced minimal Gr\"obner basis $G'$ of $\left(\begin{smallmatrix}
      \operatorname{num}(R_1) \\ \vdots \\ \operatorname{num}(R_n) \end{smallmatrix}\right)$ over $Y$ and check that $G' = G$.
  Alternatively one could verify that $\left(\begin{smallmatrix} r_1 \\ \vdots \\ r_{L(L+E)} \end{smallmatrix}\right)$ reduces to zero modulo $G'$.
  
\begin{defn} \label{defn:eta}
  A bookkeeping of this reduction yields a matrix $\eta \in Y^{{L(L+E) \times n}}$ such that
  \begin{align}
    \eta \begin{pmatrix}
      \operatorname{num}(R_1) \\
      \vdots \\
      \operatorname{num}(R_n)
    \end{pmatrix}
    =
    \begin{pmatrix}
      r_1 \\
      \vdots \\
      \vdots \\
      r_{L(L+E)}
    \end{pmatrix} \mbox{.}
  \end{align}
  We call the matrix $\eta$ \textbf{a certificate of first-order generation}.
\end{defn}
Due to \Cref{rmrk:tau_not_unique} one cannot expect that $\tau$ and $\eta$ are mutually inverse over $Y$, even when $L=1$ and $n = L(L+E)$.

\begin{rmrk} \label{rmrk:genuine}
  Different from the standard IBP relations $r_i$ and the special IBP relations in the sense of \Cref{sec:special_ibps} we will see that the numerators $\operatorname{num}(R_i) \in I^\mathrm{pol}_\mathrm{IBP} \unlhd Y^\mathrm{pol}$ of the normal-form IBP relations in the examples in \Cref{sec:examples} cannot lie in the image of the map $r: R^{L (L+E) \times 1} \to Y^\mathrm{pol}$ \eqref{eq:r}.
  The normal-form IBP relations are in this sense ``genuine''.
\end{rmrk}

\section{Sectors and symmetries of loop integrals} \label{sec:sectors}

\subsection{Sectors and scaleless integrals}

A \textbf{sector} $V$ is the set of integrals $I(z_1, \ldots, z_n)$ which have the same set of positive indices $U = \operatorname{pos}(V) = \{i | z_i > 0 \} \subseteq \{1, \ldots, n-u\}$.
This means that all integrals of a sector have the same factors $P_i$ in the denominator, possibly raised to different positive powers $z_i$.
In a given sector with subset of positive indices $U \subseteq \{1, \ldots, n-u\}$ the integral with $z_i = \chi_U(i)$ is called the \textbf{corner integral} of the sector, where $\chi_U$ is the characteristic function of $U$ as a subset of $\{1, \ldots, n\}$.
This means, the corner integral of a sector has no numerator and all factors in the denominator have power one.
A sector $V_1$ is called a \textbf{subsector} of $V_2$ if $\operatorname{pos}(V_1) \subseteq \operatorname{pos}(V_2)$.


Let $I(z_1,\ldots,z_n)$ (with nonnegative $z_i$'s) be an $L$-loop integral which does not factorize into a product of lower-loop integrals.
Setting the speed of light to unity, this integral has the physical dimension of the mass to the power $L d-2\sum z_i$.
After the integration, only masses and kinematical invariants can carry this mass dimension.
If no such quantities are available after integration, either because they are set to zero by the external kinematics or because the (evaluated) integral does not depend on them, the integral is considered to be \textbf{scaleless}.
Scaleless integrals are set to zero in dimensional regularization.
If the $L$-loop integral $I(z_1, \ldots, z_n)$ factorizes in a product of lower-loop integrals, then $I(z_1, \ldots, z_n)$ is considered scaleless if any of its factors is scaleless.

In an integral family scaleless integrals typically appear in sectors with few propagators (note that terms in the numerator cannot provide a mass scale).
It is well-known that if the corner integral of the sector is scaleless, then all integrals of this sector are scaleless as well.
In such a case we call the sector a \textbf{scaleless sector}.
Note that all subsectors of a scaleless sector are also scaleless.

For an algorithmic identification of scaleless sectors by means of linear algebra the Symanzik polynomials of the loop diagram (often denoted as $\mathcal{U}$ and $\mathcal{F}$) can be used.
Our package $\mathtt{LoopIntegrals}$ \cite{LoopIntegrals} uses ideas of \cite{Hoff2015} based on \cite{Pak_2011}, see also \cite{Lee:2013mka} for a similar algorithm.

Motivated by computations in the noncommutative rational double-shift algebra we suggest the following definition and the conjecture below:

\begin{defn}[Formally scaleless monomial] \label{defn:scaleless}
  We call a monomial of the form $D_1^{i_1} \cdots D_{n-u}^{i_{n-u}}$ ($i_j \in \N$) a \textbf{formally scaleless monomial} (with respect to $I(\underbrace{1, \ldots, 1}_{n-u}, \underbrace{0, \ldots, 0}_u)$) if
  \begin{align}
    0 = \operatorname{NF}_G(D_1^{i_1} \cdots D_n^{i_n})_{|z_1 = \ldots = z_{n-u} = 1,\, z_{n-u+1} = \ldots = z_n = 0} \mbox{.}
  \end{align}
\end{defn}


\begin{conj} \label{conj:scaleless}
  A monomial $D_1^{i_1} \cdots D_n^{i_n}$ is formally scaleless iff
  \begin{align}
    I(1-i_1, \ldots, 1-i_{n-u}, -i_{n-u+1}, \ldots, -i_n)
  \end{align}
  is scaleless.
\end{conj}

It hence suffices to look for the formally scaleless monomials among those of the form $D_1^{i_1} \cdots D_n^{i_n}$ for $i_j \in \{0,1\}$.
These correspond to the corner integrals of the respective scaleless sectors defined by the subset $\{ j \mid i_j = 1 \}$.

\subsection{Symmetries}

An integral family can have symmetries which follow from linear shifts of all momenta (loop and possibly external momenta).
Consider a linear transformation of the form
\begin{align}
 v^\prime = v \begin{pmatrix}
    M_{LL} & 0\\
    M_{EL} & M_{EE}
  \end{pmatrix}
  \label{eq:sym_trafo}
\end{align}
with $v \coloneqq ( \ell_1, \ldots , \ell_L ,k_1, \ldots ,k_E)$ and $M_{LL} \in \mathbb{Q}^{L \times L}$, $M_{EE} \in \mathbb{Q}^{E \times E}$ and $M_{EL} \in \mathbb{Q}^{E \times L}$. Such a transformation realizes a \textbf{symmetry} if
\begin{itemize}
 \item $M_{LL} \in \operatorname{GL}^{}_L(\mathbb{Q})$ and $M_{EE} \in \operatorname{GL}_E(\mathbb{Q})$;
 \item $S_i(k_1, \ldots ,k_E)$ = $S_i(k_1^\prime, \ldots ,k_E^\prime)$, i.e., all extra Lorentz invariants remain unchanged;
 \item all propagators of the transformed integral must be again in the set of allowed propagators of the integral family.
\end{itemize}
The first two conditions ensure that the transformation leaves the evaluated integral invariant (modulo a global factor given by the Jacobian $|\det M_{LL}|$).
However, the integrand might change due to the transformation and with the last condition we ensure that the transformed integrand can be expressed as linear combination of integrands of the same integral family.

Note that for the performance of Laporta's algorithm it is important to include the information on the scaleless sectors and the symmetries while generating the linear system since it reduces the number of variables and the size of the system.

\section{Examples} \label{sec:examples}

In the following examples we slightly modify the definition \eqref{eq:A} of the ring $A$ by passing to the subring
\begin{align}
  A \coloneqq \mathbb{Q}[d,m_i^2][s_1,\ldots,s_q][a_1,\ldots,a_n] \mbox{.}
\end{align}
This slightly alters the definition of $Y^\mathrm{pol}$ (but not $Y)$ accordingly.

Irreducible numerators do not occur in the examples of this section, i.e., $u=0$.

\begin{exmp}[One-loop bubble]

  The one-loop bubble is defined by the loop momentum $\ell_1$ and the external momentum $k_1$ (in particular, $L=1, E=1$).
  The two internal lines are massless.
  This results in a single independent external kinematic invariant $s = k_1^2$.
\begin{figure}[H]
  \begin{center}
    \includegraphics[width=0.3\textwidth]{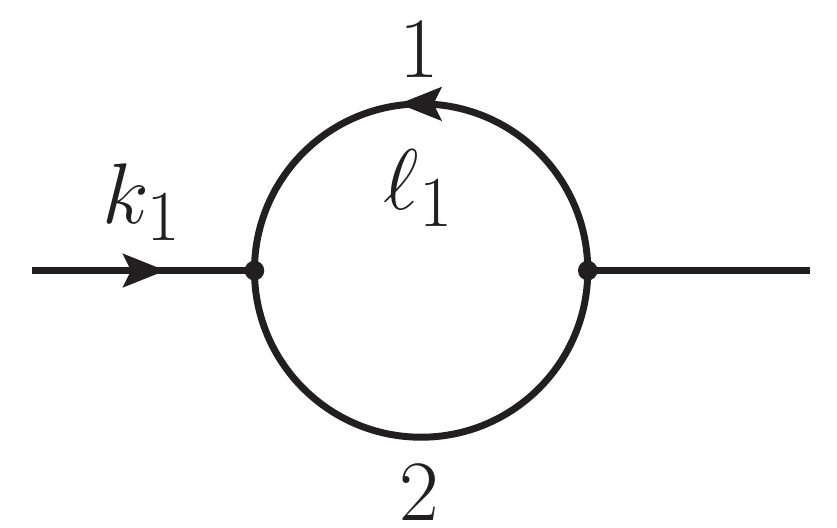}
  \end{center}
  \caption{\label{fig:olbl}The Feynman graph for the one-loop bubble integral. The arrows denote the direction of the corresponding momentum.}
\end{figure}
  The $n=2$ propagators are
  \begin{equation}
  \begin{aligned}
    P_1 &= -\ell_1^2, \\
    P_2 &= -(\ell_1+k_1)^2 \mbox{.}
  \end{aligned}
  \end{equation}
  The following transformation is a symmetry of the loop integral:
  \begin{align}
    M_1 \coloneqq \left(\begin{array}{r|r}
      -1 & \cdot \\
      \hline
      -1 & 1
    \end{array}\right) \leadsto I(z_1,z_2) = I(z_2,z_1) \mbox{,}
  \end{align}
  where the lines indicate the different block matrices in \eqref{eq:sym_trafo} and the dot denotes an entry that is equal to zero.

  The $L(L+E) = 2$ standard IBP relations are
  \begin{equation}
  \begin{aligned}
    r_1
    &= -a_2 D_1 D_2^- - s a_2 D_2^- + (d - 2 a_1 - a_2) , \\
    r_2
    &= -a_1 D_1^- D_2 + a_2 D_1 D_2^- - s a_1 D_1^- + s a_2 D_2^- + (a_1 - a_2 ) \mbox{,}
  \end{aligned}
  \end{equation}
  expressed as elements of the polynomial double-shift algebra
  \begin{align}
    Y^\mathrm{pol} \coloneqq \Q[d,s][a_1,a_2] \langle D_i, D_i^- \mid i = 1, 2 \rangle \mbox{.}
  \end{align}
  
  The reduced Gr\"obner basis $G^\mathrm{pol}$ for the one-loop bubble over the polynomial double-shift algebra $Y^\mathrm{pol}$ has $4$ elements and was computed in less than a second using \textsc{Plural}:
    \allowdisplaybreaks
    \begin{align}
    G^\mathrm{pol} = \bigg\{ &
    (a_2 - 1) D_1 - (d - 2 a_1 - a_2 + 1) D_2 + (a_2 - 1) s, \nonumber\\
  & -(d - a_1 - 2 a_2 + 1) D_1 + (a_1 - 1) D_2 + (a_1 - 1) s, \nonumber\\
  & a_1 (d - 2 a_1 - 2) D_1^- - a_2 (d - 2 a_2 - 2) D_2^-, \nonumber\\
  & (a_2 - 1) (d - 2 a_2) D_1 - (a_1 - 1) (d - 2 a_1) D_2 \bigg\} \mbox{.}
    \end{align}
  
  The reduced Gr\"obner basis over the \emph{rational} double-shift algebra
  \begin{align}
    Y \coloneqq \Q(d,s)(a_1,a_2) \langle D_i, D_i^- \mid i = 1, 2 \rangle
  \end{align}
  was computed in less than a second using $\mathtt{Ore\_algebra}$.
  It also has 4 elements and reads
    \allowdisplaybreaks
    \begin{align}
    G
    = \bigg\{
    & (d-a_1-a_2) (d-2 a_1-2 a_2+2) D_2 - (a_2-1) s (d-2 a_2) , \nonumber \\[0.2em]
    & (d-a_1-a_2) (d-2 a_1-2 a_2+2) D_1 - (a_1-1) s (d-2 a_1) , \nonumber \\[0.2em]
    & a_2  s (d-2 a_2-2) D_2^{-} - (d-a_1-a_2-1) (d-2 a_1-2 a_2), \nonumber \\[0.2em]
    & a_1  s (d-2 a_1-2) D_1^{-} - (d-a_1-a_2-1) (d-2 a_1-2 a_2) \bigg\} \mbox{.} \label{eq:GB1loopbubble}
  \end{align}

We can now compute the normal forms of the operators $a_i D_i^-$ with respect to the Gr\"obner basis $G$ of the left ideal
  \begin{align}
    I_\mathrm{IBP}
    \coloneqq
    \langle r_1, r_2 \rangle \lhd Y
  \end{align}
  generated by the above two standard IBP relations:
  \begin{equation} \label{eq:one-loop-bubble-nf-minus}
  \begin{aligned}
    \operatorname{NF}_G(a_1 D_1^-) =
    & \frac{(d-a_1-a_2-1)(d-2a_1-2a_2)}{(d-2a_1-2) s},
    \\
    \operatorname{NF}_G(a_2 D_2^-) =
    & \frac{(d-a_1-a_2-1)(d-2a_1-2a_2)}{(d-2a_2-2) s}\mbox{.}
  \end{aligned}
  \end{equation}
  The $C_2$-symmetry of the problem is obvious from the above normal forms.
  It further manifests itself in the normal forms of the indeterminates $D_i$:
  {\small
  \begin{equation} \label{eq:one-loop-bubble-nf-plus}
  \begin{aligned}
    \operatorname{NF}_G(1) =\,
    &
    1
    &&\text{$\leadsto 1$ is (trivially) a standard monomial},
    \\
    \operatorname{NF}_G(D_1) =\,
    &
    \frac{(a_1-1)(d-2a_1)s}{(d-a_1-a_2)(d-2a_1-2a_2+2)}
    &&\text{$\leadsto D_1$ is a nonstandard monomial},
    \\
    \operatorname{NF}_G(D_2) =\,
    &
    \frac{(a_2-1)(d-2a_2)s}{(d-a_1-a_2)(d-2a_1-2a_2+2)}
    &&\text{$\leadsto D_2$ is a nonstandard monomial.}
  \end{aligned}
  \end{equation}
  }%
  
  The set of standard monomials with respect to $G$ is merely
  \begin{align}
    \{ 1 \} \mbox{,}
  \end{align}
  which by \Cref{rmrk:standard_monomials} corresponds to the single master integral
  \begin{align} \label{eq:masters_oneloopbubble}
    \{ I(1,1) \} \mbox{.}
  \end{align}
  
Computing normal forms with respect to the Gr\"obner basis $G$ one can easily verify that the minimal scaleless monomials $D_1$, $D_2$ are indeed formally scaleless with respect to $I(1,1)$ in the sense of \Cref{defn:scaleless}.

  As discussed in \Cref{sec:Groebner} the normal form of $a_i D_i^-$ with respect to the \emph{polynomial} Gr\"obner basis $G^\mathrm{pol} \subset Y^\mathrm{pol}$ includes expressions in the $D_j^-$'s.
  Still, the polynomial reduction verifies that for $i=1,2$ the numerator $\operatorname{num}(R_i) \in Y^\mathrm{pol}$ of $R_i \coloneqq a_i D_i^- - \operatorname{NF}_G(a_i D_i^-) \in I_\mathrm{IBP} \lhd Y$ reduces to zero modulo $G^\mathrm{pol}$ in $Y^\mathrm{pol}$, proving the inclusion $\{ \operatorname{num}(R_i) \mid i = 1, 2\} \subset I_\mathrm{IBP}^\mathrm{pol}$.
  These reductions yield as in \Cref{defn:polynomiality} the certificate of polynomiality matrix
  {\small
  \begin{align}
  \tau
  =
  \begin{pmatrix*}[r]
  s (a_1 D_1^- -a_2 D_2^-) -a_2 D_1 D_2^- -d+a_1+2 a_2+1 &-s a_2 D_2^- -a_2 D_1 D_2^- \\
-a_2 D_1 D_2^- -d +a_1+2 a_2+1 & -a_2 D_1 D_2^-
  \end{pmatrix*} \in (Y^\mathrm{pol})^{2 \times 2}
  \end{align}}%
  satisfying
  \begin{align} \label{eq:tau_bubble}
    N \coloneqq \begin{pmatrix}
      \operatorname{num}(R_1) \\
      \operatorname{num}(R_2)
    \end{pmatrix}
    =
    \tau
    \begin{pmatrix}
      r_1 \\
      r_2
    \end{pmatrix} \mbox{.}
  \end{align}
  The following nontrivial row-syzygy
  {\small
  \begin{align}
    \begin{pmatrix}
      s a_1 D_1^- + a_1 D_2 D_1^- -s a_2 D_2^- - a_2 D_1 D_2^- - a_1 + a_2 &
      -s a_2 D_2^- - a_2 D_1 D_2^- + d -2 a_1 -a_2-1
    \end{pmatrix}
  \end{align}}
  in $(Y^\mathrm{pol})^{1 \times 2}$ of $\left(\begin{smallmatrix} r_1 \\ r_2 \end{smallmatrix}\right)$ shows that \eqref{eq:tau_bubble} does not uniquely determine $\tau$.
  
  The computation of the Gr\"obner basis of $N$ in $Y^\mathrm{pol}$ verifies that $N$ does not generate $I_\mathrm{IBP}^\mathrm{pol}$.
  However, $r_1, r_2$ reduce to zero modulo $N$ in $Y$, yielding the certificate matrix
  {\footnotesize
  \begin{align}
    \eta
    :=
    \frac{1}{(d-a_1-a_2-1)(d-2 a_1-2 a_2)}
    \begin{pmatrix*}[r]
      a_2 D_1 D_2^- & -d + 2 a_1 + a_2 \\
      -d + 2 a_1 + a_2 + a_2 D_1 D_2^- & a_1 D_2 D_1^- + d - 2 a_1 - a_2
    \end{pmatrix*} \in Y^{2 \times 2} \mbox{,}
  \end{align}}
  of first-order generation from \Cref{defn:eta}, proving that $\langle R_1, R_2 \rangle_Y = I_\mathrm{IBP} \coloneqq \langle r_1, r_2 \rangle_Y$

  As mentioned in \Cref{rmrk:genuine}, the numerators of the normal-form IBP relations $R_1, R_2$ cannot lie in the image of the map $r$, since it is obvious from \eqref{eq:IBP} that the IBP relations in the image of $r$ are affine expressions in the $a_i$'s.
  
We can now verify that the normal forms in \eqref{eq:one-loop-bubble-nf-minus} and \eqref{eq:one-loop-bubble-nf-plus} reflect the functional dependence of the integral's closed-form result on the indices $z_1$ and $z_2$, analogously to \eqref{eq:contiguous_one-loop-tadpole}.
Using
\begin{equation}
\begin{aligned}
  I(z_1, z_2)
  &=
    \int \operatorname{d}^d \ell_1 \frac{1}{\left(-\ell_1^2 \right)^{z_1} \left(-(\ell_1 + k_1)^2 \right)^{z_2}} \\
  &=
    i\pi^{d/2}\,
    \frac{
    \Gamma\left(\frac{d}{2} - z_1\right)\,
    \Gamma\left(\frac{d}{2} - z_2\right)\,
    \Gamma\left(z_1 + z_2 - \frac{d}{2}\right)}{
    \Gamma(z_1)\,
    \Gamma(z_2)\,
    \Gamma(d - z_1 - z_2)\,
    (-s)^{z_1 + z_2 - d/2}}
\end{aligned}
\end{equation}
and $\Gamma(z+1)=z\,\Gamma(z)$ we have, for example, the \emph{contiguous function relation}
\begin{align} \label{eq:contiguous_bubble}
  z_1\, I(z_1 + 1, z_2) =
  \frac{(d-z_1-z_2-1)(d-2z_1-2z_2)}{(d-2z_1-2) s}\, I(z_1, z_2)
  \mbox{,}
\end{align}
which corresponds to $\operatorname{NF}_G(a_1 D_1^-)$ in \eqref{eq:one-loop-bubble-nf-minus}.
  
  Note that 1 is always a standard monomial, unless the left ideal is the unit left ideal.
  This happens for example for the \emph{scaleless} one-loop bubble (with $s=k_1^2=0$).
  There $G^\mathrm{pol}$ has five elements but $G = \{1\}$, i.e., $I_\mathrm{IBP} = Y$.
This means that the set of standard monomials is empty ($\operatorname{NF}_G(1) = 0$), i.e., there are no master integrals as expected for a scaleless integral family (contrary to what is claimed in \cite{Gerdt:2004kt}).

\end{exmp}

\begin{exmp}[One-loop box] \label{exmp:one-loopbox}

  The one-loop box is defined by the loop momentum $\ell_1$ and linearly dependent external momenta $k_1, k_2, k_3, k_4$.
  Due to momentum conservation $k_1+k_2+k_3+k_4 = 0$ we take $k_1, k_2, k_4$ to be the set of linearly independent external momenta (in particular, $L=1, E=3$).
  The external lines are on-shell and massless implying $k_i^2=0$ for $i=1,2,3,4$.
  Internal lines are also massless.
  This results in the independent external kinematic invariants $s_{12} = 2 k_1 \cdot k_2$ and $s_{14} = 2 k_1 \cdot k_4$.
\begin{figure}[H]
  \begin{center}
    \includegraphics[width=0.3\textwidth]{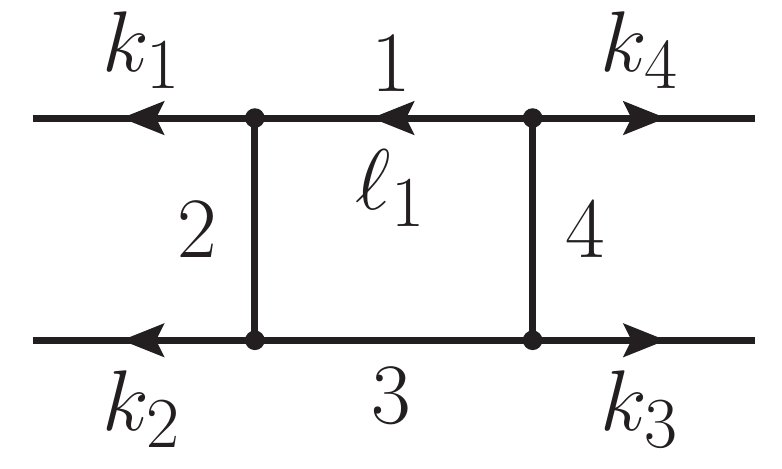}
  \end{center}
  \caption{\label{fig:olb}The Feynman graph for the one-loop box integral. The arrows denote the direction of the corresponding momentum.}
\end{figure}

  The $n=4$ propagators are
  \begin{equation}
  \begin{aligned}
    P_1 &= -\ell_1^2, \\
    P_2 &= -(\ell_1-k_1)^2, \\
    P_3 &= -(\ell_1 - k_1 - k_2)^2, \\
    P_4 &= -(\ell_1 + k_4)^2 \mbox{.}
  \end{aligned}
  \end{equation}
  The following two involutions generate a Kleinian symmetry group $V_4$ of the loop integral:
  \begin{align}
    M_1 \coloneqq \left(\begin{array}{r|rrr}
      -1 & \cdot & \cdot & \cdot \\
      \hline
      \cdot & \cdot & -1 & 1 \\
      \cdot & \cdot & -1 & \cdot \\
      \cdot & 1 & -1 & \cdot
    \end{array}\right) \leadsto I(z_1,z_2,z_3,z_4) = I(z_1,z_4,z_3,z_2),
    \\
    M_2 \coloneqq \left(\begin{array}{r|rrr}
      -1 & \cdot & \cdot & \cdot \\
      \hline
      1 & \cdot & 1 & -1 \\
      1 & 1 & \cdot & -1 \\
      \cdot & \cdot & \cdot & -1
    \end{array}\right) \leadsto I(z_1,z_2,z_3,z_4) = I(z_3,z_2,z_1,z_4) \mbox{.}
  \end{align}
  Their product yields
  \begin{align}
     M_1 M_2 \eqqcolon M_3 = \left(\begin{array}{r|rrr}
      1 & \cdot & \cdot & \cdot \\
      \hline
      -1 & -1 & \cdot & \cdot \\
      -1 & -1 & \cdot & 1 \\
      \cdot & -1 & 1 & \cdot
    \end{array}\right) \leadsto I(z_1,z_2,z_3,z_4) = I(z_3,z_4,z_1,z_2) \mbox{.}
  \end{align}
  
  For completeness, we mention that additional symmetry relations can be found when some of the indices are negative, e.g., for $z_1 \leq 0,z_3 \leq 0$ the matrix
  \begin{align}
    \left(\begin{array}{r|rrr}
      -1 & \cdot & \cdot & \cdot \\
      \hline
      1 & 1 & \cdot & \cdot \\
      \cdot & \cdot & 1 & \cdot \\
      -1 & \cdot & \cdot & 1
    \end{array}\right)
  \end{align}
  encodes the symmetry
  \begin{align}
    I(z_1,z_2,z_3,z_4) = I(0,z_4,0,z_2) \bullet (s_{14}-D_1+D_2+D_4)^{-z_1} (s_{14}+D_2-D_3+D_4)^{-z_3} \mbox{.}
  \end{align}
  Note that in this example this symmetry is realized in a simpler form by $M_1$.
  
  \medskip
  The $L(L+E) = 4$ standard IBP relations are
  \begin{equation}
  {\tiny
  \begin{aligned}
    r_1
    &= -a_2 D_1 D_2^- -a_3 D_1 D_3^- -a_4 D_1 D_4^- -s_{12} a_3 D_3^- + (d-2a_1-a_2-a_3-a_4), \\
    r_2
    &= a_1 D_1^- D_2 - a_2 D_1 D_2^- -a_3 D_1 D_3^- + a_3 D_2 D_3^- - a_4 D_1 D_4^- + a_4 D_2 D_4^- -s_{12} a_3 D_3^- + s_{14} a_4 D_4^- -a_1+a_2, \\
    r_3
    &= - a_1 D_1^- D_2 + a_1 D_1^- D_3 + a_2 D_2^- D_3 - a_3 D_2 D_3^- - a_4 D_2 D_4^- + a_4 D_3 D_4^- + s_{12} a_1 D_1^- - s_{14} a_4 D_4^- -a_2 + a_3, \\
    r_4
    &= a_2 D_1 D_2^- + a_3 D_1 D_3^- - a_1 D_1^- D_4 - a_2 D_2^- D_4 - a_3 D_3^- D_4 + a_4 D_1 D_4^- - s_{14} a_2 D_2^- + s_{12} a_3 D_3^- +a_1-a_4 \mbox{,}
  \end{aligned}
  }
  \end{equation}
  expressed as elements of the polynomial double-shift algebra
  \begin{align}
    Y^\mathrm{pol} \coloneqq \Q[d,s_{12},s_{14}][a_1,a_2,a_3,a_4] \langle D_i, D_i^- \mid i = 1, \ldots, 4 \rangle \mbox{.}
  \end{align}
  
  The reduced Gr\"obner basis $G^\mathrm{pol}$ for the one-loop box over the polynomial double-shift algebra $Y^\mathrm{pol}$ has $28$ elements and was computed in less than a second using \textsc{Plural}.
  The reduced Gr\"obner basis over the \emph{rational} double-shift algebra
  \begin{align}
    Y \coloneqq \Q(d,s_{12},s_{14})(a_1,a_2,a_3,a_4) \langle D_i, D_i^- \mid i = 1, \ldots, 4 \rangle
  \end{align}
  was computed in less than five seconds using $\mathtt{Ore\_algebra}$.
  It has 9 elements and reads
  {\footnotesize
    \allowdisplaybreaks
    \begin{align}
    G&
    = \bigg\{ D_4-D_2+\frac{(a_2-a_4) s_{14}}{d-a_{1234}}, D_3-D_1+\frac{(a_1-a_3) s_{12}}{d-a_{1234}}, \nonumber \\[0.4em]
 & 4 (a_2-1) (d-a_{1234}) D_3-2 (d-2a_{134}) (d-a_{1234}) D_4 +(d-2a_{14}-2) (d-2a_{234}) s_{12}  \nonumber\\[0.2em]
 & -2 (d-2a_{134}) (a_2-a_4) s_{14} -\frac{(d-2a_{14}-2) (d-2a_{34}-2) a_4 s_{12} s_{14}}{d-a_{1234}-1} \, D_4^{-}, \nonumber \\[0.4em]
 & -2 (d-2a_{234}) (d-a_{1234}) D_3+4 (a_1-1) (d-a_{1234}) D_4-2 (a_1-a_3) (d-2a_{234}) s_{12} \nonumber\\[0.2em]
 & +(d-2a_{23}-2) (d-2a_{134}) s_{14} -\frac{(d-2a_{23}-2) a_3 (d-2a_{34}-2) s_{12} s_{14}}{d-a_{1234}-1} \, D_3^{-}, \nonumber \\[0.4em]
 & 4 (d-a_{1234}) (a_4-1) D_3-2 (d-2a_{123}) (d-a_{1234}) D_4\nonumber\\[0.2em]
 & +(d-2a_{12}-2) (d-2a_{234}) s_{12} -\frac{(d-2a_{12}-2) a_2 (d-2a_{23}-2) s_{12} s_{14}}{d-a_{1234}-1} \, D_2^{-}, \nonumber \\[0.4em]
 & -2 (d-2a_{124}) (d-a_{1234}) D_3+4 (a_3-1) (d-a_{1234}) D_4\nonumber\\[0.2em]
 &+(d-2a_{12}-2) (d-2a_{134}) s_{14} -\frac{a_1 (d-2a_{12}-2) (d-2a_{14}-2) s_{12} s_{14}}{d-a_{1234}-1} \, D_1^{-}, \nonumber \\[0.4em]
 & -2 (d-2a_{1234}+4) (d-a_{1234}+1) D_4^2 +(d-2a_{124}+2) (d-2a_{234}+2) s_{12} D_4\nonumber\\[0.2em]
 &-2 (d-2a_{1234}+4) (a_2-a_4+1) s_{14} D_4+4 (a_2-1) (a_4-1) s_{14} D_3\nonumber\\[0.2em]
 & -\frac{(d-2a_{124}+2) (d-2a_{34}) (a_4-1) s_{12} s_{14}}{d-a_{1234}}, \nonumber \\[0.4em]
 & -(d-2a_{1234}+4) (d-a_{1234}+1) D_3 D_4 +(a_3-1) (d-2a_{234}+2) s_{12} D_4 \nonumber\\[0.2em]
 &+(d-2a_{134}+2) (a_4-1) s_{14} D_3 -\frac{(a_3-1) (d-2a_{34}) (a_4-1) s_{12} s_{14}}{d-a_{1234}}, \nonumber \\[0.4em]
 & -2 (d-2a_{1234}+4) (d-a_{1234}+1) D_3^2 +(d-2a_{123}+2) (d-2a_{134}+2) s_{14} D_3\nonumber\\[0.2em]
 & - 2 (a_1-a_3+1) (d-2a_{1234}+4) s_{12} D_3+4 (a_1-1) (a_3-1) s_{12} D_4 \nonumber\\[0.2em]
 &-\frac{(d-2a_{123}+2) (a_3-1) (d-2a_{34}) s_{12} s_{14}}{d-a_{1234}}\bigg\} \mbox{,}
  \end{align}}%
  with the abbreviations $a_{i_1 \ldots i_k} \coloneqq \sum_{j=1}^k a_{i_j}$.
  We provide both $G$ and $G^\mathrm{pol}$  in computer-readable form in the ancillary files.
  
  We can now compute the normal forms of the operators $a_i D_i^-$ with respect to the Gr\"obner basis $G$ of the left ideal
  \begin{align}
    I_\mathrm{IBP}
    \coloneqq
    \langle r_i \mid i = 1, \ldots, 4 \rangle \lhd Y
  \end{align}
  generated by the above four standard IBP relations:
  {\scriptsize
  \begin{align}
    \operatorname{NF}_G(a_1 D_1^-) =
    &
    -\frac{2\left(d-2a_{124}\right) \left(d-a_{1234}\right) \left(d-a_{1234}-1\right)}{\left(d-2a_{12}-2\right) \left(d-2a_{14}-2\right) s_{12} s_{14}} D_3 + \frac{4 \left(a_3-1\right) \left(d-a_{1234}\right) \left(d-a_{1234}-1\right)}{\left(d-2a_{12}-2\right) \left(d-2a_{14}-2\right) s_{12} s_{14}} D_4 \notag \\
   & + \frac{\left(d-2a_{134}\right) \left(d-a_{1234}-1\right)}{\left(d-2a_{14}-2\right) s_{12}},
    \notag \\
    \operatorname{NF}_G(a_2 D_2^-) =
    &
    \frac{4 \left(a_4-1\right) \left(d-a_{1234}\right) \left(d-a_{1234}-1\right)}{\left(d-2a_{12}-2\right)
   \left(d-2a_{23}-2\right) s_{12} s_{14}} D_3 -\frac{2 \left(d-2a_{123}\right) \left(d-a_{1234}\right) \left(d-a_{1234}-1\right)}{\left(d-2a_{12}-2\right) \left(d-2a_{23}-2\right) s_{12} s_{14}} D_4 \notag \\
   &+ \frac{\left(d-2a_{234}\right) \left(d-a_{1234}-1\right)}{\left(d-2a_{23}-2\right) s_{14}},
    \notag \\
    \operatorname{NF}_G(a_3 D_3^-) =
    &
    -\frac{2 \left(d-2a_{234}\right) \left(d-a_{1234}\right) \left(d-a_{1234}-1\right)}{\left(d-2 a_{23}-2\right) \left(d-2a_{34}-2\right) s_{12} s_{14}} D_3
    +\frac{4 \left(a_1-1\right) \left(d-a_{1234}\right) \left(d-a_{1234}-1\right)}{\left(d-2a_{23}-2\right) \left(d-2a_{34}-2\right) s_{12} s_{14}} D_4
    \notag \\
    &+\frac{\left(d-2a_{134}\right) \left(d-a_{1234}-1\right)}{\left(d-2a_{34}-2\right) s_{12}}-\frac{2 \left(a_1-a_3\right)
   \left(d-2a_{234}\right) \left(d-a_{1234}-1\right)}{\left(d-2a_{23}-2\right) \left(d-2a_{34}-2\right)
   s_{14}},
    \notag \\
    \operatorname{NF}_G(a_4 D_4^-) =
    &
    \frac{4 \left(a_2-1\right) \left(d-a_{1234}\right) \left(d-a_{1234}-1\right)}{\left(d-2a_{14}-2\right) \left(d-2a_{34}-2\right) s_{12} s_{14}} D_3
   -\frac{2 \left(d-2a_{134}\right) \left(d-a_{1234}\right) \left(d-a_{1234}-1\right)}{\left(d-2a_{14}-2\right) \left(d-2a_{34}-2\right) s_{12} s_{14}} D_4 \notag \\
    &
    +\frac{\left(d-2a_{234}\right) \left(d-a_{1234}-1\right)}{\left(d-2a_{34}-2\right) s_{14}}-\frac{2 \left(a_2-a_4\right) \left(d-2a_{134}\right) \left(d-a_{1234}-1\right)}{\left(d-2a_{14}-2\right) \left(d-2a_{34}-2\right) s_{12}} \mbox{.} \label{eq:one-loop-box}
    \end{align}
  }%
  
  The normal forms of the indeterminates $D_i$ reveal the $V_4$-symmetry of the problem:
  \begin{align}
    \operatorname{NF}_G(D_1) =\,
    &
    D_3+\frac{(a_1-a_3)s_{12}}{d-a_{1234}}
    &&\text{$\leadsto D_1$ is a nonstandard monomial}
    \notag
    \\
    \operatorname{NF}_G(D_2) =\,
    &
    D_4+\frac{(a_2-a_4)s_{14}}{d-a_{1234}}
    &&\text{$\leadsto D_2$ is a nonstandard monomial}
    \notag
    \\
    \operatorname{NF}_G(D_3) =\,
    &
    D_3
    &&\text{$\leadsto D_3$ is a standard monomial}
    \\
    \operatorname{NF}_G(D_4) =\,
    &
    D_4
    &&\text{$\leadsto D_4$ is a standard monomial}
    \notag
    \\
    \operatorname{NF}_G(1) =\,
    &
    1
    &&\text{$\leadsto 1$ is (trivially) a standard monomial}
    \notag
  \end{align}
  
  The set of standard monomials with respect to $G$ is
  \begin{align}
    \{ 1, D_3, D_4 \} \mbox{,}
  \end{align}
  which by \Cref{rmrk:standard_monomials} correspond to the three master integrals
  \begin{align} \label{eq:masters_oneloopbox}
    \{ I(1,1,1,1), I(1,1,0,1), I(1,1,1,0) \} \mbox{.}
  \end{align}

  Computing normal forms with respect to the Gr\"obner basis $G$ one can easily verify that the minimal scaleless monomials $D_1 D_2$, $D_1 D_4$, $D_2 D_3$, $D_3 D_4$ are indeed formally scaleless with respect to $I(1,1,1,1)$ in the sense of \Cref{defn:scaleless}.
  The above computations can be found in the notebook \cite{1LoopBox}.
  
  As discussed in \Cref{sec:Groebner} the normal form of $a_i D_i^-$ with respect to the \emph{polynomial} Gr\"obner basis $G^\mathrm{pol} \subset Y^\mathrm{pol}$ includes expressions in the $D_j^-$'s.
  Still, the polynomial reduction verifies that for $i=1,\ldots,4$ the numerator $\operatorname{num}(R_i) \in Y^\mathrm{pol}$ of $R_i \coloneqq a_i D_i^- - \operatorname{NF}_G(a_i D_i^-) \in I_\mathrm{IBP} \lhd Y$ reduces to zero modulo $G^\mathrm{pol}$ in $Y^\mathrm{pol}$, proving the inclusion $\{ \operatorname{num}(R_i) \mid i = 1, \ldots, 4\} \subset I_\mathrm{IBP}^\mathrm{pol}$.
  These reductions yield, as in \Cref{defn:polynomiality}, the certificate of polynomiality matrix $\tau \in (Y^\mathrm{pol})^{4 \times 4}$.
  The matrix $\tau$ is included in the digital paper supplements.
    The following nontrivial syzygy relation $f_1 r_1 + f_4 r_4 = 0$ of $\left(\begin{smallmatrix} r_1 \\ \vdots \\ r_4 \end{smallmatrix}\right)$ with
  \begin{equation}
  {\scriptsize
  \begin{aligned}
    f_1 =&
    s_{14} a_2 D_2^- -a_2 D_1 D_2^- -s_{12} a_3 D_3^- -a_3 D_1 D_3^- +a_1 D_4 D_1^- +a_2 D_4 D_2^ - +a_3 D_4 D_3^- -a_4 D_1 D_4^- -a_1+a_4,
    \\
    f_4 =&
    -a_2 D_1 D_2^- -s_{12} a_3 D_3^- -a_3 D_1 D_3^- -a_4 D_1 D_4^- +d-2 a_1-a_2-a_3-a_4-1
  \end{aligned}}%
  \end{equation}
shows that the matrix $\tau$ is not uniquely determined.
Furthermore, we have verified that the (normalized) reduced minimal Gr\"obner bases of $\{r_1, \ldots, r_4\}$ and $\{ \operatorname{num}(R_1), \ldots, \operatorname{num}(R_4) \}$ over $Y$ coincide, proving that the latter (or equivalently $\{ R_1, \ldots, R_4 \}$) generates the left ideal $I_\mathrm{IBP} \lhd Y$.
  
  The rational Gr\"obner basis $G$ can be used to perform fast reductions.
  The ancillary files of the arXiv submission include a \textsc{FORM}~\cite{Ruijl:2017dtg} program and a \textsc{Mathematica} program which compute normal forms modulo $G$.
  In order to express the integral $I(z_1,\ldots, z_4)$ in terms of the above three master integrals \eqref{eq:masters_oneloopbox} the program computes the normal form of the monomial $(D_1^-)^{z_1} (D_2^-)^{z_2} (D_3^-)^{z_3} (D_4^-)^{z_4} \in Y$ modulo $G$.
  For example, the provided \textsc{FORM}-program is able to express $I(10,10,10,10)$ in terms of the master integrals in about 5 seconds.
  
  We also provide a \textsc{Mathematica} program which uses the four normal-form IBPs $\{ R_i \coloneqq a_i D_i^- - \operatorname{NF}_G(a_i D_i^-) \mid i = 1, \ldots, 4 \} \subset I_\mathrm{IBP} \lhd Y$ to perform the reduction.
This program illustrates that the normal form IBPs can be used to construct an algorithm for the reduction to master integrals similarly to the standard IBPs.
However, due to the much simpler structure of the normal form IBPs, it is much easier to design the reduction algorithm.
In particular, they already have the correct form to reduce integrals in the top-level sector, where all indices $z_i$ with $i\in\{1,\ldots,n-u\}=\{1,\ldots,4-0\}$ are positive.

\end{exmp}

\begin{exmp}[Two-loop tadpole with three massive lines] \label{exmp:2LoopTadpole}

  The two-loop tadpole with three massive lines is defined by the loop momenta $\ell_1,\ell_2$ (in particular, $L=2, E=0$).
  The three internal lines are massive with equal mass $m$.
\begin{figure}[H]
  \begin{center}
    \includegraphics[width=0.2\textwidth]{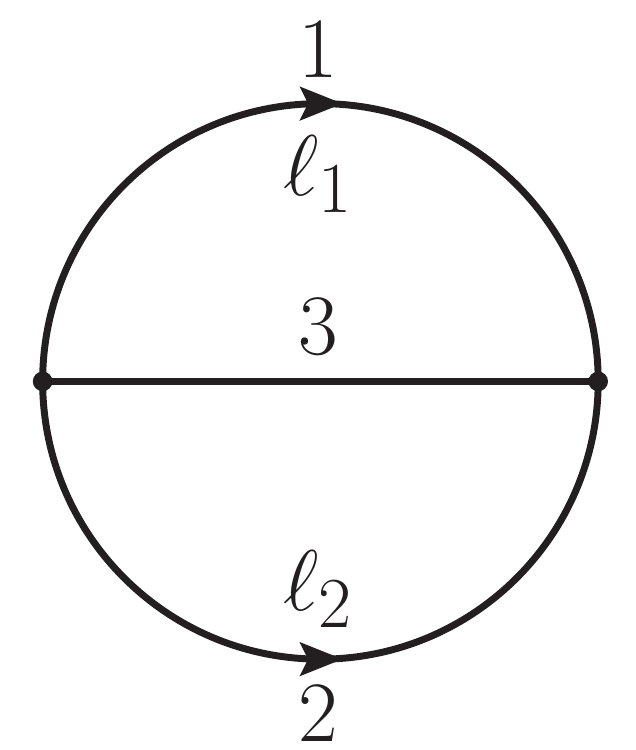}
  \end{center}
  \caption{\label{fig:tlt}The Feynman graph for the two-loop tadpole integral with three massive lines. The arrows denote the direction of the corresponding momentum.}
\end{figure}
  The $n=3$ propagators are
  \begin{equation}
  \begin{aligned}
    P_1 &= -\ell_1^2+m^2, \\
    P_2 &= -\ell_2^2+m^2, \\
    P_3 &= -(\ell_1+\ell_2)^2+m^2 \mbox{.}
  \end{aligned}
  \end{equation}
  The symmetric group $S_3$ is a symmetry group with the following three transpositions as generators:
  \begin{align}
    M_1 \coloneqq
    \left(\begin{array}{rr|}
      \cdot & 1 \\
      1 & \cdot \\
      \hline
    \end{array}\right)
    \leadsto I(z_1,z_2,z_3) = I(z_2,z_1,z_3),
    \\
    M_2 \coloneqq
    \left(\begin{array}{rr|}
      -1 & \cdot \\
      -1 & 1 \\
      \hline
    \end{array}\right)
    \leadsto I(z_1,z_2,z_3) = I(z_3,z_2,z_1),
    \\
    M_3 \coloneqq
    \left(\begin{array}{rr|}
      1 & -1 \\
      \cdot & -1 \\
      \hline
    \end{array}\right)
    \leadsto I(z_1,z_2,z_3) = I(z_1,z_3,z_2) \mbox{.}
  \end{align}
  The $L(L+E) = 4$ standard IBP relations are
  \begin{equation}
  \begin{aligned}
    r_1
    &= 2 m^2 a_1 D_1^ - +m^2 a_3 D_3^- -a_3 D_1 D_3^ - +a_3 D_2 D_3^ - +d-2 a_1-a_3 , \\
    r_2
    &= -m^2 a_1 D_1^ - +m^2 a_3 D_3^ - +a_1 D_1^- D_2-a_1 D_1^- D_3+a_3 D_1 D_3^- -a_3 D_2 D_3^ - +a_1-a_3, \\
    r_3
    &= -m^2 a_2 D_2^ - +m^2 a_3 D_3^ - +a_2 D_1 D_2^- -a_2 D_2^- D_3-a_3 D_1 D_3^ - +a_3 D_2 D_3^ - +a_2-a_3, \\
    r_4
    &= 2 m^2 a_2 D_2^ - +m^2 a_3 D_3^ - +a_3 D_1 D_3^- -a_3 D_2 D_3^ - +d-2 a_2-a_3 \mbox{,}
  \end{aligned}
  \end{equation}
  expressed as elements of the polynomial double-shift algebra
  \begin{align}
    Y^\mathrm{pol} \coloneqq \Q[d,m][a_1,a_2,a_3] \langle D_i, D_i^- \mid i = 1, 2, 3 \rangle \mbox{.}
  \end{align}
  
  The reduced Gr\"obner basis $G^\mathrm{pol}$ for the two-loop tadpole with three massive lines over the polynomial double-shift algebra $Y^\mathrm{pol}$ has $47$ elements and was computed in a few seconds using \textsc{Plural}.  
  The reduced Gr\"obner basis over the \emph{rational} double-shift algebra
  \begin{align}
    Y \coloneqq \Q(d,m)(a_1,a_2,a_3) \langle D_i, D_i^- \mid i = 1, 2, 3 \rangle
  \end{align}
  was also computed in a few seconds using $\mathtt{Ore\_algebra}$ and has 9 elements.
  Both $G$ and $G^\mathrm{pol}$ are included in the ancillary files.

We can now compute the normal forms of the operators $a_i D_i^-$ with respect to the Gr\"obner basis $G$ of the left ideal
  \begin{align}
    I_\mathrm{IBP}
    \coloneqq
    \langle r_1, \ldots, r_4 \rangle \lhd Y
  \end{align}
  generated by the above four standard IBP relations:
  {\scriptsize
  \begin{align}
    &\operatorname{NF}_G(a_1 D_1^-) = \notag \\
    &\frac{(-2 d+2 a_1+3 a_2+3 a_3-4) D_1+(d-a_2-3 a_3+2) D_2+(d-3 a_2-a_3+2) D_3-m^2 (2 d-4 a_1-a_2-a_3)}{6m^2} \mbox{,}
    \notag \\
    &\operatorname{NF}_G(a_2 D_2^-) = \notag \\
    &\frac{(d-a_1-3 a_3+2) D_1+(-2 d+3 a_1+2 a_2+3 a_3-4) D_2+(d-3 a_1-a_3+2) D_3-m^2 (2 d-a_1-4 a_2-a_3)}{6m^2} \mbox{,}
    \notag \\
    &\operatorname{NF}_G(a_3 D_3^-) = \notag \\
    &\frac{(d-a_1-3 a_2+2) D_1+(d-3 a_1-a_2+2) D_2+(-2 d+3 a_1+3 a_2+2 a_3-4) D_3-m^2 (2 d-a_1-a_2-4 a_3)}{6m^2} \mbox{.}\notag \\
  \end{align}}
  The $S_3$-symmetry of the problem is obvious from the above normal forms.
  {\small
  \begin{equation}
  \begin{aligned}
    \operatorname{NF}_G(1) =\,
    &
    1
    &&\text{$\leadsto 1$ is (trivially) a standard monomial},
    \\
    \operatorname{NF}_G(D_1) =\,
    &
    D_1
    &&\text{$\leadsto D_1$ is a standard monomial},
    \\
    \operatorname{NF}_G(D_2) =\,
    &
    D_2
    &&\text{$\leadsto D_2$ is a standard monomial},
    \\
    \operatorname{NF}_G(D_2) =\,
    &
    D_3
    &&\text{$\leadsto D_3$ is a standard monomial.}
  \end{aligned}
  \end{equation}
  }%
  
  The set of standard monomials with respect to $G$ is
  \begin{align}
    \{ 1, D_1, D_2, D_3 \} \mbox{,}
  \end{align}
  which by \Cref{rmrk:standard_monomials} corresponds to the master integrals
  \begin{align} \label{eq:masters_massivetadpole}
    \{ I(1,1,1), I(0,1,1), I(1,0,1), I(1,1,0) \}
  \end{align}
  of which the last three are equal.
  It is interesting to note that the homogeneity of the integral with respect to $m^2$ manifests itself through the relation:
  \begin{align}
    -m^2 \operatorname{NF}_{G^\mathrm{pol}}(a_1 D_1^- + a_2 D_2^- + a_3 D_3^-)
    =
    -m^2 \operatorname{NF}_G(a_1 D_1^- + a_2 D_2^- + a_3 D_3^-)
    =
    d-a_1-a_2-a_3 \mbox{,}
  \end{align}
  or equivalently the IBP operator
  \begin{align}
    -m^2(a_1 D_1^- + a_2 D_2^- + a_3 D_3^-) - (d-a_1-a_2-a_3) \in I_\mathrm{IBP}^\mathrm{pol} \mbox{,}
  \end{align}
  where
  \begin{align}
    I(z_1,z_2,z_3) \bullet \left(-m^2(a_1 D_1^- + a_2 D_2^- + a_3 D_3^-)\right)
    =
    m^2 \frac{\partial}{\partial m^2} I(z_1,z_2,z_3) \mbox{.}
  \end{align}
  
  Computing normal forms with respect to the Gr\"obner basis $G$ one can easily verify that the minimal scaleless monomials $D_1 D_2$, $D_1 D_3$, $D_2 D_3$ are indeed formally scaleless with respect to $I(1,1,1)$ in the sense of \Cref{defn:scaleless}.

  As discussed in \Cref{sec:Groebner} the normal form of $a_i D_i^-$ with respect to the \emph{polynomial} Gr\"obner basis $G^\mathrm{pol} \subset Y^\mathrm{pol}$ includes expressions in the $D_j^-$'s.
  Still, the polynomial reduction verifies that for $i=1,2,3$ the numerator $\operatorname{num}(R_i) \in Y^\mathrm{pol}$ of $R_i \coloneqq a_i D_i^- - \operatorname{NF}_G(a_i D_i^-) \in I_\mathrm{IBP} \lhd Y$ reduces to zero modulo $G^\mathrm{pol}$ in $Y^\mathrm{pol}$, proving the inclusion $\{ \operatorname{num}(R_i) \mid i = 1, 2, 3\} \subset I_\mathrm{IBP}^\mathrm{pol}$.
  These reductions yield as in \Cref{defn:polynomiality} the certificate of polynomiality matrix
  \begin{align}
  \tau
  =
  \begin{pmatrix}
  3 m^2+D_1 & 0 & -2 m^2-D_2 &  -m^2+D_1-D_2-D_3 \\
-2 D_1 & -3 D_1 & -2 m^2+2 D_2 & 2 m^2+D_1+2 D_2-D_3 \\
D_1 & 3 D_1 & 4 m^2-D_2 & 2 m^2-2 D_1-D_2+2 D_3
  \end{pmatrix} \in (Y^\mathrm{pol})^{3 \times 4}
  \end{align}
  satisfying
  \begin{align} \label{eq:tau_massivetadpole}
    N \coloneqq \begin{pmatrix}
      \operatorname{num}(R_1) \\
      \operatorname{num}(R_2) \\
      \operatorname{num}(R_3)
    \end{pmatrix}
    =
    \tau
    \begin{pmatrix}
      r_1 \\
      \vdots \\
      r_4
    \end{pmatrix} \mbox{.}
  \end{align}
  The following nontrivial row-syzygy
  {\small
  \begin{align}
    \begin{pmatrix}
      -a_3 D_3^- & -a_2 D_2^- -a_3 D_3^- & a_1 D_1^ - +a_3 D_3^- & a_3 D_3^-
    \end{pmatrix} \in (Y^\mathrm{pol})^{1 \times 4}
  \end{align}}
  of $\left(\begin{smallmatrix} r_1 \\ \vdots \\ r_4 \end{smallmatrix}\right)$ shows that \eqref{eq:tau_massivetadpole} does not uniquely determine $\tau$.
  
  The computation of the Gr\"obner basis of $N$ in $Y^\mathrm{pol}$ verifies that $N$ does not generate $I_\mathrm{IBP}^\mathrm{pol}$.
  However, we have verified that the (normalized) reduced minimal Gr\"obner bases of $\{r_1, \ldots, r_4\}$ and $\{ \operatorname{num}(R_1), \operatorname{num}(R_2), \operatorname{num}(R_3) \}$ over $Y$ coincide, proving that the latter (or equivalently $\{ R_1, R_2, R_3 \}$) generates the left ideal $I_\mathrm{IBP} \lhd Y$.

\end{exmp}

\section{The special IBP relations} \label{sec:special_ibps}

In the previous section we computed the set of normal-form IBP relations $\{ R_i \mid i = 1, \ldots, n \}$ starting from the set of standard IBP relations $\{ r_i \mid i = 1, \ldots, L(L+E) \}$.
However, it turns out to be more efficient to start with a different set of IBP relations introduced in \cite{Gluza:2010ws,Kosower:2018obg}.

Recall that the standard IBP relations are obtained from \eqref{eq:IBP}
\begin{align}
  r(C) \coloneqq d \cdot C^i_i + \operatorname{tr} \mathcal{E}_C -a_c D_c^- \mathcal{E}_{j,c}^i C^j_i \in Y^\mathrm{pol} \subset Y
\end{align}
when $C \in R^{L(L+E) \times 1}$ runs through a standard basis $\{e_1, \ldots, e_{L(L+E)}\}$, where $\mathcal{E} \in R^{n \times L(L+E)}$ is the IBP-generating matrix over the polynomial ring $R \coloneqq \mathbb{F}[s_1,\ldots, s_q][D_1, \ldots, D_n] \subset Y^\mathrm{pol} \subset Y$.
The ansatz of \cite{Kosower:2018obg} is to find vectors $C \in R^{L(L+E) \times 1}$, possibly different from the standard basis vectors, such that $r(C)$ does not include $D_c^-$ for $c \in \{1, \ldots, n-u \}$.
This means that for
\begin{align}
  \mathcal{E'} \coloneqq \mathcal{E}_{c=1,\ldots,n-u} \in R^{(n-u) \times L(L+E)}
\end{align}
the $c$-th row of $\mathcal{E}' C$ must be a multiple of $D_c$ for $c \in \{1, \ldots, n-u \}$.
This can be easily achieved by computing over the polynomial ring $R$ a column syzygies matrix of the matrix $\mathcal{E}'$ modulo the diagonal matrix $\operatorname{diag}(D_1, \ldots, D_{n-u}) \in R^{(n-u) \times (n-u)}$.

\begin{defn}
  A \textbf{column syzygies matrix} of a matrix $A \in R^{r \times c_A}$ is a matrix $S \in R^{c_A \times ?}$, such that
  \begin{enumerate}
    \item $A S = 0$;
    \item if $A T = 0$ then there exists a matrix $Z$ with $T = S Z$.
  \end{enumerate}
  A \textbf{column syzygies matrix} of a matrix $A \in R^{r \times c_A}$ \textbf{modulo} a matrix $B \in R^{r \times c_B}$ is a matrix $S \in R^{c_A \times ?}$, such that
  \begin{enumerate}
    \item $A S = 0$ modulo $B$, i.e., there exists a matrix $X$ with $A S = B X$;
    \item if $A T = 0$ modulo $B$ then there exists a matrix $Z$ with $T = S Z$.
  \end{enumerate}
  
\end{defn}

\begin{rmrk}
  A column syzygies matrix $S \in R^{c_A \times c_S}$ of $A \in R^{r \times c_A}$ modulo $B \in R^{r \times c_B}$ can be computed by computing a column syzygies matrix $\left(\frac{S}{T}\right) \in R^{(c_A + c_B) \times c_S}$ of the augmented matrix $(A|B) \in R^{r \times (c_A + c_B)}$ and then discarding the lower part $T$.
  Most computer algebra systems with a Gr\"obner basis engine provide procedures to compute $S$ directly without (fully) computing $T$.
\end{rmrk}

\begin{defn}
  Let $\mathcal{S}$ be a column syzygies matrix of $\mathcal{E}'$ modulo $\operatorname{diag}(D_1, \ldots, D_{n-u})$.
  Further let $\langle \mathcal{S} \rangle$ be the $R$-submodule of $R^{L(L+E) \times 1}$ generated by the columns of $\mathcal{S}$.
  We call
  \begin{align}
    r\left(\langle \mathcal{S} \rangle\right) \coloneqq \{ r(C) \mid C \text{ a column of } \langle \mathcal{S} \rangle \}
  \end{align}
  the set of \textbf{special IBP relations} of the loop diagram.
\end{defn}

\begin{rmrk}
  In all examples where we were able to compute the (reduced) Gr\"obner bases of the standard IBP relations and the special IBP relations in the noncommutative polynomial double-shift algebra $Y^\mathrm{pol}$ they coincided, i.e., we were able to prove that
  \begin{align}
    \langle r(C) \mid C \text{ a column of } \mathcal{S} \rangle_{Y^\mathrm{pol}} = I_\mathrm{IBP}^\mathrm{pol} = \langle r_i \mid i=1,\ldots, {L(L+E)} \rangle_{Y^\mathrm{pol}} \mbox{.}
  \end{align}
  This means, the set of special IBP relations generates the same left ideal in $Y^\mathrm{pol}$ as the standard IBP relations.
  This is remarkable since $\langle \mathcal{S} \rangle$ is a \emph{proper} $R$-submodule of $R^{L(L+E) \times 1}$.
\end{rmrk}

\begin{exmp}[One-loop box, \Cref{exmp:one-loopbox}, continued]
  For the one-loop box we have the IBP-generating matrix (cf.~\eqref{eq:Euler})
  {\footnotesize
  \begin{align}
    \mathcal{E}' = \mathcal{E} &=
    \left( \begin{array}{cccc}
 2D_{1} & D_{1}-D_{2} & -s_{12}+D_{2}-D_{3} & -D_{1}+D_{4} \\
 D_{1}+D_{2} & D_{1}-D_{2} & D_{2}-D_{3} & s_{14}-D_{1}+D_{4} \\
 s_{12}+D_{1}+D_{3} & s_{12}+D_{1}-D_{2} & D_{2}-D_{3} & -s_{12}-D_{1}+D_{4} \\
 D_{1}+D_{4} & -s_{14}+D_{1}-D_{2} & s_{14}+D_{2}-D_{3} & -D_{1}+D_{4}
\end{array} \right) \in R^{4 \times 4} \mbox{.}
  \end{align}}%
  The reduced column syzygy matrix of $\mathcal{E}' \in R^{4 \times 4}$ modulo $\operatorname{diag}(D_1, D_2, D_3, D_4) \in R^{4 \times 4}$ is
  \begin{align}
    \mathcal{S} \coloneqq {\scriptsize\left( \begin{array}{cccccc}
 D_{2}-D_{4} & D_{1}-D_{3} & (s_{12}+2D_{3}-2D_{4})D_{4} & (s_{14}-2D_{3}+2 D_{4})D_{3} & -D_{3}D_{4}^{2} & D_{3}^{2}D_{4} \\
 D_{4} & -D_{1} & (-s_{12}+D_{2}-D_{3}+2D_{4})D_{4} & -(D_{1}+D_{4})D_{3} & D_{3}D_{4}^{2} & \cdot \\
 \cdot & -D_{1} & (D_{1}+D_{2})D_{4} & -2D_{1}D_{3} & \cdot & D_{1}D_{3}D_{4} \\
 D_{2} & \cdot & 2D_{2}D_{4} & -(D_{1}+D_{2})D_{3} & D_{2}D_{3}D_{4} & \cdot
\end{array} \right)} \mbox{,}
  \end{align}
  in $R^{4 \times 6}$ producing $6$ special IBP relations in $Y^\mathrm{pol} \subset Y$.
  They are provided in an ancillary file attached to the arXiv submission.
  As mentioned above, a Gr\"obner basis computation in $Y^\mathrm{pol}$ verifies that the $6$ special IBP relations generate $I_\mathrm{IBP}^\mathrm{pol} \lhd Y^\mathrm{pol}$.
\end{exmp}

\section{The Linear Algebra Ansatz} \label{sec:LA-Ansatz}

We now describe a method for producing the normal-form IBP relations without computing the Gr\"obner basis for $I_\mathrm{IBP} \lhd Y$.
This method is based on linear algebra (LA-Ansatz).

We start with a generating set $M = \{ g_1, \ldots, g_t \}$ of the left ideal $I_\mathrm{IBP} \lhd Y$, e.g., the set of standard IBP relations or the set of special IBP relations.
For a fixed order $o \geq 0$:
\begin{itemize}
  \item Multiply $M$ by $(D_1^-)^{i_1} \cdots (D_n^-)^{i_n}$ for $i_j \geq 0$ and $i_1 + \cdots + i_n \leq o$ from the left.
  \item Compute the matrix of coefficients of the resulting  set of elements in $I_\mathrm{IBP} \lhd Y$ as a matrix over
  \begin{align}
    A \coloneqq \mathbb{Q}[d,m_i^2][s_1, \ldots, s_q][a_1, \ldots, a_n] \subset K \coloneqq \mathbb{F}(s_1, \ldots, s_q)(a_1, \ldots, a_n) \mbox{.}
  \end{align}
  \item Use Gaussian elimination over the field $K$ to find expressions of monomials in the $D_i^-$'s in terms of polynomials in the $D_j$'s.
\end{itemize}

If the order $o$ is too small one cannot expect to express sufficiently many monomials in the $D_i^-$'s in terms of polynomials in the $D_j$'s.

\begin{rmrk}
Our general experience shows that starting with the standard IBP relations requires higher values for the order $o$ than starting with the special IBP relations.
\end{rmrk}

One can simulate the Gaussian elimination over $K$ by running the algorithm over the subring $A$ and dividing by the content of the resulting rows of reduction steps.
This relies on fast multivariate gcd computations for which we use the \textsf{Julia} package \textsc{Hecke}.
The complement of the subring $\mathbb{Q}[a_1,\ldots,a_n] < A$ is a multiplicative subset of $A$.
Strictly speaking one must carry out the elimination in the localized subring $(A \setminus \mathbb{Q}[a_1, \ldots, a_n])^{-1} A < K$.
We leave this for future work.

\begin{figure}
  \begin{center}
    \includegraphics[width=0.4\textwidth]{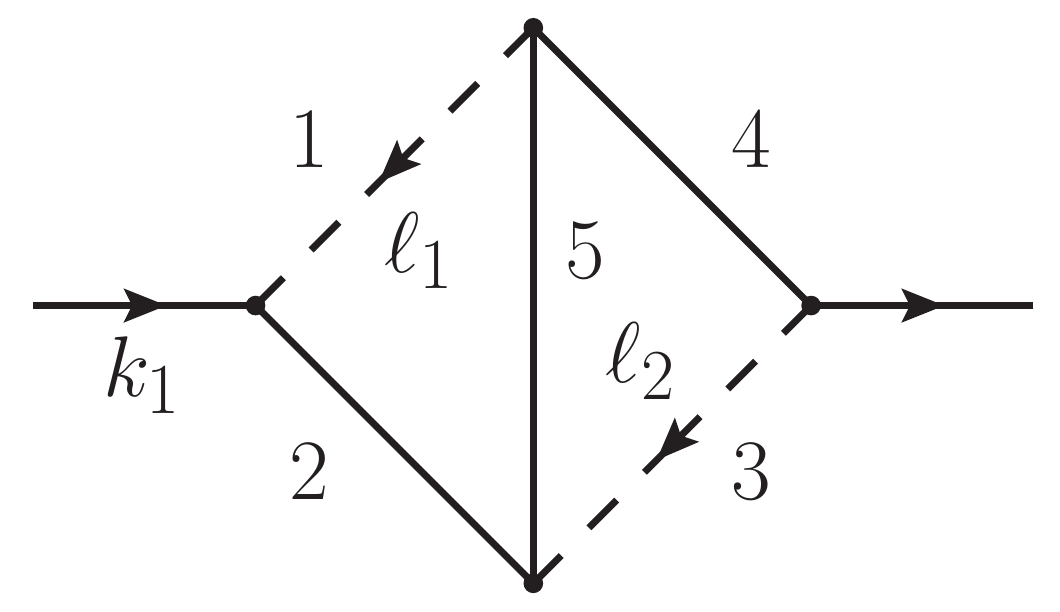}
  \end{center}
  \caption{\label{fig:on-shell-kite}The Feynman graph for the on-shell kite
    integral. Dashed lines denote massless propagators and solid lines
    denote massive propagators with mass $m$. The arrows denote the
    direction of the corresponding momentum.}
\end{figure}

\begin{exmp}[One-loop box]
  The LA-Ansatz applied to the $6$ special IBP relations of \Cref{exmp:one-loopbox} leads already for order $o = 2$ to exactly the same first-order normal-form IBP relations $\{ R_i \coloneqq a_i D_i^- - \operatorname{NF}_G(a_i D_i^-) \mid i = 1, \ldots, 4 \}$ as produced by the GB-Ansatz (see \eqref{eq:one-loop-box}).
  The LA-Ansatz simulates the Gaussian elimination over the field $K = \mathbb{Q}(d)(s_{12},s_{14})(a_1,\ldots,a_4)$.
\end{exmp}

\begin{exmp}[On-shell kite]
  The on-shell kite is defined by the loop momenta $\ell_1, \ell_2$ and the external momentum $k_1$ as indicated in \Cref{fig:on-shell-kite}.
  Therefore the only independent kinematic invariant is $s \coloneqq k_1^2$.
  Owing to the on-shell condition we have $s=m^2$.
  The propagators are
  \begin{equation}
  \begin{aligned}
    P_1 &= -\ell_1^2, \\
    P_2 &= m^2-(\ell_1+k_1)^2, \\
    P_3 &= m^2-(\ell_2+k_1)^2, \\
    P_4 &= -\ell_2^2, \\
    P_5 &= m^2-(\ell_1+\ell_2+k_1)^2 \mbox{.}
  \end{aligned}
  \end{equation}
  
  For the on-shell kite we have the IBP-generating matrix (cf.~\eqref{eq:Euler})
  {\footnotesize
  \begin{eqnarray}
    \mathcal{E}' &=& \mathcal{E} =
    \left( \begin{array}{ccc}
2 D_1 & -D_2-D_3+D_5 & -D_1+D_2 \\
D_1+D_2 & -D_2-D_4+D_5 & -2 s-D_1+D_2 \\
0 & 0 & 0 \\
0 & 0 & 0 \\
D_1-D_3+D_5& -D_2+D_4+D_5 & -2 s-D_1+D_2+D_3-D_4
\end{array} \right.
  \nonumber\\
  && \hspace{2.5cm} \left. \begin{array}{ccc}
  0 & 0 & 0 \\
  0 & 0&   0 \\
  -D_1-D_3+D_5 & D_3+D_4 & -2 s+D_3-D_4 \\
  -D_2-D_3+D_5 & 2 D_4 & D_3-D_4 \\
  D_1-D_3+D_5 & -D_2+D_4+D_5 & -2 s-D_1+D_2+D_3-D_4
\end{array} \right) \in R^{5 \times 6} \mbox{.}
  \end{eqnarray}}%
  The reduced column syzygy matrix $\mathcal{S}$ of $\mathcal{E}' \in R^{6 \times 13}$ modulo $\operatorname{diag}(D_1, D_2, D_3, D_4, D_5) \in R^{5 \times 5}$ is producing $13$ special IBP relations in $Y^\mathrm{pol} \subset Y$.
  The matrix $\mathcal{S}$ and the corresponding special IBP relations are provided in ancillary files attached to the arXiv submission.
  Since we were not able to compute $G^\mathrm{pol}$ we were not able to verify that the $13$ special IBP relations generate $I_\mathrm{IBP}^\mathrm{pol} \lhd Y^\mathrm{pol}$.
  
  Applying the LA-Ansatz to the 13 special IBP relations we get
  \begin{equation}
  \begin{aligned}
    R_1 = a_1 D_1^-
    &
    -\left(\frac{p_{10}}{4 d_{1} d_{2} d_{3} d_{4} d_{7} d_{8} s}
    +
    \frac{p_{12} D_2 + p_{13} D_3 + p_{14} D_4 + p_{15} D_5}{16 d_{1} d_{2} d_{3} d_{4} d_{7} d_{8} d_{9} s^2} \right),
    \\
    R_2 = a_2 D_2^-
    &
    -\left(\frac{p_{20}}{8 d_{1} d_{3} d_{7} d_{10} s}
    +
    \frac{p_{22} D_2 + p_{23} D_3 + p_{24} D_4 + p_{25} D_5}{32 d_{1} d_{3} d_{7} d_{9} d_{10} s^2} \right),
    \\
    R_3 = a_3 D_3^-
    &
    -\left(\frac{p_{30}}{8 d_{1} d_{5} d_{7} d_{10} s}
    +
    \frac{p_{32} D_2 + p_{33} D_3 + p_{34} D_4 + p_{35} D_5}{32 d_{1} d_{5} d_{7} d_{9} d_{10} s^2} \right),
    \\
    R_4 = a_4 D_4^-
    &
    -\left(\frac{p_{40}}{4 d_{1} d_{2} d_{5} d_{6} d_{10} d_{11} s}
    +
    \frac{p_{42} D_2 + p_{43} D_3 + p_{44} D_4 + p_{45} D_5}{16 d_{1} d_{2} d_{5} d_{6} d_{9} d_{10} d_{11} s^2} \right),
    \\
    R_5 = a_5 D_5^-
    &
    -\left(\frac{p_{50}}{8 d_{1} d_{7} d_{10} s}
    +
    \frac{p_{52} D_2 + p_{53} D_3 + p_{54} D_4 + p_{55} D_5}{32 d_{1} d_{7} d_{9} d_{10} s^2} \right) \mbox{.}
  \end{aligned}%
  \end{equation}
  with 
  \begin{align}
  \begin{array}{lll}
    d_1 &\coloneqq a_{1123445}-2 d+1
    &\coloneqq
    2 a_{1}+a_{2}+a_{3}+2 a_{4}+a_{5}-2 d+1, \\
    d_2 &\coloneqq a_{1123445}-2 d+2
    &\coloneqq
    2a_{1}+a_{2}+a_{3}+2 a_{4}+a_{5}-2 d+2, \\
    d_3 &\coloneqq a_{112}-d+1
    &\coloneqq
    2a_{1}+a_{2}-d+1, \\
    d_4 &\coloneqq a_{112}-d+2
    &\coloneqq
    2 a_{1}+a_{2}-d+2, \\
    d_5 &\coloneqq a_{344}-d+1
    &\coloneqq
    a_{3}+2 a_{4}-d+1, \\
    d_6 &\coloneqq a_{344}-d+2
    &\coloneqq
    a_{3}+2 a_{4}-d+2, \\
    d_7 &\coloneqq a_{11235}-d
    &\coloneqq
    2a_{1}+a_{2}+a_{3}+a_{5}-d, \\
    d_8 &\coloneqq a_{11235}-d+1
    &\coloneqq
    2a_{1}+a_{2}+a_{3}+a_{5}-d+1, \\
    d_9 &\coloneqq a_{23445}-d-1
    &\coloneqq
    a_{2}+a_{3}+2 a_{4}+a_{5}-d-1, \\
    d_{10} &\coloneqq a_{23445}-d
    &\coloneqq
    a_{2}+a_{3}+2a_{4}+a_{5}-d,\\
    d_{11} &\coloneqq a_{23445}-d+1
    &\coloneqq
    a_{2}+a_{3}+2 a_{4}+a_{5}-d+1 \mbox{.}
  \end{array}
  \end{align}
  The 13 special IBP relations and the above 5 normal-form IBPs produced by the LA-Ansatz are provided in electronic form as ancillary files.
\end{exmp}

\section{Conclusion} \label{sec:conclusion}

In this paper we propose the (noncommutative) \emph{rational} double-shift algebra\footnote{Another possibility is to use the \emph{rational} Weyl algebra. The (polynomial) Weyl algebra appears in \cite{Lee:2008tj}.} $Y$ \eqref{eq:Y} as an algebraic structure in which the IBP relations form a left ideal $I_\mathrm{IBP}$ (\Cref{defn:I_IBP}).
This algebra admits a Gr\"obner basis notion which we exploit for effective computations.
We proved in \Cref{prop:finite_generation} that this left ideal is finitely generated by the standard IBP relations $r_1, \ldots, r_{L(L+E)}$ \eqref{eq:standard_IBPs}.

Motivated by the observation that the standard IBP relation of the one-loop tadpole forms a reduced Gr\"obner basis of the left ideal $I_\mathrm{IBP}$ we introduced the notion of first-order normal-form IBP relations $R_1, \ldots, R_n$ (\Cref{defn:1st-order_normal-form}), which for nontrivial examples have a significantly more structured form than the standard IBP relations $r_1, \ldots, r_{L(L+E)}$.
The obvious connection between the $R_i$'s and contiguous function relations as in \eqref{eq:contiguous_one-loop-tadpole} and \eqref{eq:contiguous_bubble} suggests their mathematical significance.
Furthermore, the examples in \Cref{sec:examples} show how the symmetries of the integral family are manifested in the normal-form IBP relations.

The first-order normal-form IBP relations in turn led us to the definition of a first-order integral family (\Cref{defn:1storder}).
These are special families for which $R_1, \ldots, R_n$ form a generating set of $I_\mathrm{IBP}$.
We verified for the examples treated in \Cref{sec:examples} that they are first-order integral families.

An obvious advantage of the Gr\"obner basis $G$ and the normal-form IBP relations $R_1, \ldots, R_n$ ---as opposed to the Laporta algorithm--- is that they need to be computed \emph{only once} (and stored in a database) and can then be applied \emph{individually} to any integral of the corresponding integral family.
The latter makes this approach well-suited for parallelization.
However, the bottleneck of this approach is the computation of the noncommutative Gr\"obner basis $G$ of $I_\mathrm{IBP} \unlhd Y$.
Indeed, the existing implementations failed to produce the Gr\"obner basis for the on-shell kite.

The bottleneck in computing the Gr\"obner basis $G$ motivated the Linear Algebra Ansatz developed in \Cref{sec:LA-Ansatz} which enabled us to compute the first-order normal-form IBP relations even when we were not able to compute $G$.
Here it turned out that the special IBP relations of \Cref{sec:special_ibps} yield a more efficient set of generators of $I_\mathrm{IBP}$ than the standard IBP relations.
The special IBP relations are in this sense intermediate between the standard and the normal-form IBP relations.

Examples in \Cref{sec:examples} clearly demonstrate, why the normal-form IBP relations ---regardless of their definition using Gr\"obner bases--- are ideally suited for IBP reductions.
We have provided as an ancillary file a proof of concept \textsc{Mathematica} program which uses the normal-form IBP relations to perform reductions for the integral family of the one-loop box.

Preliminary experiments show that the use of normal-form IBP relations in Laporta's algorithm improves the sparsity of the system of linear equations.
We leave the systematic combination of both approaches for future work.

While we focused on first-order integral families in the work, we expect that integral families with more than one master integral per sector are ``higher-order'' generated.
It would be interesting to have a simple characterization for first-order generation.
We also leave this for future work.

Finally, we mention that the main ideas of Sections~\ref{sec:Y} and~\ref{sec:Groebner} have also been discussed in a somewhat more pedagogical form in~\cite{Barakat:2022ttc}.

\section*{Acknowledgments}

This research was supported by the Deutsche Forschungsgemeinschaft (DFG,
German Research Foundation) under grant 396021762 -- TRR 257 ``Particle Physics Phenomenology after the Higgs Discovery.''

\appendix

\section{Mathematical details} \label{sec:App}

In this Appendix we give the mathematical constructions of the algebras $T$, $\widetilde{T}$ and $R$ from \Cref{sec:Y}.

Let $\mathbb{F}[\ell_1, \ldots, \ell_L, k_1, \ldots, k_E]$ be the polynomial algebra in the $(L+E)d'$ indeterminates $\ell_1^\mu, \ldots, \ell_L^\mu, k_1^\mu, \ldots, k_E^\mu$ ($\mu = 0, \ldots, d'-1$) with coefficients in the field $\mathbb{F} \coloneqq \Q(d, m_i^2)$.
Then define the residue class algebra
\begin{align}
  \widetilde{T} \coloneqq \mathbb{F}[\ell_1, \ldots, \ell_L, k_1, \ldots, k_E] / \langle \rho_o \mid o = 1, \ldots, r \rangle
\end{align}
where $\rho_o$ are affine polynomials in the Lorentz invariant quadratic expressions $k_i \cdot k_j$ (called relations of external momenta).
The elements of $\widetilde{T}$ are residue classes of the form
\begin{align}
  f + \langle \rho_o \mid o = 1, \ldots, r \rangle \mbox{,}
\end{align}
where $f \in \mathbb{F}[\ell_1, \ldots, \ell_L, k_1, \ldots, k_E]$ and $\langle \rho_o \mid o = 1, \ldots, r \rangle$ is the ideal generated by $\rho_o$, for $o = 1, \ldots, r$.

Consider the subalgebra
\begin{align}
  T = \mathbb{F}[p_i \cdot p_j \mid p \in \{\ell, k\}, i,j = 1, \ldots, L+E]  \leq \widetilde{T}
\end{align}
generated by the Lorentz invariant quadratic expressions $p_i \cdot p_j$ (for $p \in \{ \ell, k \}$), where we silently identify $\ell_i, k_j$ with their residue classes in $\widetilde{T}$.
The $n = {\frac{L(L+1)}{2}+LE}$ propagators
\begin{align}
  P_1, \ldots, P_n
\end{align}
are $\mathbb{F}$-linearly independent elements of the $(q+n)$-dimensional $\mathbb{F}$-linear subspace $T_2 < T$ generated by the quadratic expressions $p_i \cdot p_j$ (for $p \in \{ \ell, k \}$), where $q$ depends on the relations $\rho_o$.
This means, that (the representative in $\widetilde{T}$ of) each $P_c$ is an affine polynomial in the quadratic expressions $p_i \cdot p_j$ with constant term either a $\Q$-multiple of $m_i^2 \in \mathbb{F}$ or a rational number.
Let
\begin{align}
  S_1, \ldots, S_q, P_1, \ldots, P_n
\end{align}
be a basis of $T_2$.
The representatives of $S_1, \ldots, S_q$ in the (ambient) algebra $\mathbb{F}[\ell_1, \ldots, \ell_L, k_1, \ldots, k_E]$ can be chosen homogeneous of degree $2$ and are called the \textbf{extra Lorentz invariants}.
Together with the masses they form the kinematic invariants of the process.
Then $T$ can equally be expressed as the subring
\begin{align}
  T = \mathbb{F}[S_1, \ldots, S_q][P_1, \ldots, P_n] \leq \widetilde{T} \mbox{.}
\end{align}
It is isomorphic to the polynomial $\mathbb{F}$-algebra
\begin{align}
  R \coloneqq \mathbb{F}[s_1,\ldots, s_q][D_1, \ldots, D_n]
\end{align}
under the polynomial embedding
\begin{align}
  \lambda:
  \begin{cases}
    R &\hookrightarrow \widetilde{T}, \\
    D_c &\mapsto P_c, \\
    s_e &\mapsto S_e \mbox{,}
  \end{cases}
\end{align}
having $T$ as its image in $\widetilde{T}$.

\providecommand{\href}[2]{#2}\begingroup\raggedright\endgroup


\begin{thebibliography}{10}

\bibitem{HHJP}
A.~Huss, J.~Huston, S.~Jones and M.~Pellen, \emph{Les houches 2021: Physics at
  tev colliders: Report on the standard model precision wishlist},  2022.
\newblock 10.48550/ARXIV.2207.02122.

\bibitem{Weinzierl:2022eaz}
S.~Weinzierl, \emph{{Feynman Integrals}}, Springer Cham (2022),
  \href{https://doi.org/10.1007/978-3-030-99558-4}{10.1007/978-3-030-99558-4},
  [\href{https://arxiv.org/abs/2201.03593}{{\ttfamily 2201.03593}}].

\bibitem{Tkachov:1981wb}
F.V.~Tkachov, \emph{{A Theorem on Analytical Calculability of Four Loop
  Renormalization Group Functions}},
  \href{https://doi.org/10.1016/0370-2693(81)90288-4}{\emph{Phys. Lett. B}
  {\bfseries 100} (1981) 65}.

\bibitem{Chetyrkin:1981qh}
K.G.~Chetyrkin and F.V.~Tkachov, \emph{{Integration by Parts: The Algorithm to
  Calculate beta Functions in 4 Loops}},
  \href{https://doi.org/10.1016/0550-3213(81)90199-1}{\emph{Nucl. Phys. B}
  {\bfseries 192} (1981) 159}.

\bibitem{Laporta:2000dsw}
S.~Laporta, \emph{{High precision calculation of multiloop Feynman integrals by
  difference equations}},
  \href{https://doi.org/10.1142/S0217751X00002159}{\emph{Int. J. Mod. Phys. A}
  {\bfseries 15} (2000) 5087}
  [\href{https://arxiv.org/abs/hep-ph/0102033}{{\ttfamily hep-ph/0102033}}].

\bibitem{vonManteuffel:2014ixa}
A.~von Manteuffel and R.M.~Schabinger, \emph{{A novel approach to integration
  by parts reduction}},
  \href{https://doi.org/10.1016/j.physletb.2015.03.029}{\emph{Phys. Lett. B}
  {\bfseries 744} (2015) 101}
  [\href{https://arxiv.org/abs/1406.4513}{{\ttfamily 1406.4513}}].

\bibitem{Peraro:2016wsq}
T.~Peraro, \emph{{Scattering amplitudes over finite fields and multivariate
  functional reconstruction}},
  \href{https://doi.org/10.1007/JHEP12(2016)030}{\emph{JHEP} {\bfseries 12}
  (2016) 030} [\href{https://arxiv.org/abs/1608.01902}{{\ttfamily
  1608.01902}}].

\bibitem{Anastasiou:2004vj}
C.~Anastasiou and A.~Lazopoulos, \emph{{Automatic integral reduction for higher
  order perturbative calculations}},
  \href{https://doi.org/10.1088/1126-6708/2004/07/046}{\emph{JHEP} {\bfseries
  07} (2004) 046} [\href{https://arxiv.org/abs/hep-ph/0404258}{{\ttfamily
  hep-ph/0404258}}].

\bibitem{Smirnov:2008iw}
A.V.~Smirnov, \emph{{Algorithm FIRE -- Feynman Integral REduction}},
  \href{https://doi.org/10.1088/1126-6708/2008/10/107}{\emph{JHEP} {\bfseries
  10} (2008) 107} [\href{https://arxiv.org/abs/0807.3243}{{\ttfamily
  0807.3243}}].

\bibitem{Smirnov:2014hma}
A.V.~Smirnov, \emph{{FIRE5: a C++ implementation of Feynman Integral
  REduction}}, \href{https://doi.org/10.1016/j.cpc.2014.11.024}{\emph{Comput.
  Phys. Commun.} {\bfseries 189} (2015) 182}
  [\href{https://arxiv.org/abs/1408.2372}{{\ttfamily 1408.2372}}].

\bibitem{Smirnov:2019qkx}
A.V.~Smirnov and F.S.~Chuharev, \emph{{FIRE6: Feynman Integral REduction with
  Modular Arithmetic}},
  \href{https://doi.org/10.1016/j.cpc.2019.106877}{\emph{Comput. Phys. Commun.}
  {\bfseries 247} (2020) 106877}
  [\href{https://arxiv.org/abs/1901.07808}{{\ttfamily 1901.07808}}].

\bibitem{Studerus:2009ye}
C.~Studerus, \emph{{Reduze-Feynman Integral Reduction in C++}},
  \href{https://doi.org/10.1016/j.cpc.2010.03.012}{\emph{Comput. Phys. Commun.}
  {\bfseries 181} (2010) 1293}
  [\href{https://arxiv.org/abs/0912.2546}{{\ttfamily 0912.2546}}].

\bibitem{vonManteuffel:2012np}
A.~von Manteuffel and C.~Studerus, \emph{{Reduze 2 - Distributed Feynman
  Integral Reduction}},  1, 2012.

\bibitem{Maierhofer:2017gsa}
P.~Maierh\"ofer, J.~Usovitsch and P.~Uwer, \emph{{Kira\textemdash{}A Feynman
  integral reduction program}},
  \href{https://doi.org/10.1016/j.cpc.2018.04.012}{\emph{Comput. Phys. Commun.}
  {\bfseries 230} (2018) 99}
  [\href{https://arxiv.org/abs/1705.05610}{{\ttfamily 1705.05610}}].

\bibitem{Klappert:2020nbg}
J.~Klappert, F.~Lange, P.~Maierh\"ofer and J.~Usovitsch, \emph{{Integral
  reduction with Kira 2.0 and finite field methods}},
  \href{https://doi.org/10.1016/j.cpc.2021.108024}{\emph{Comput. Phys. Commun.}
  {\bfseries 266} (2021) 108024}
  [\href{https://arxiv.org/abs/2008.06494}{{\ttfamily 2008.06494}}].

\bibitem{Peraro_2019}
T.~Peraro, \emph{{FiniteFlow}: multivariate functional reconstruction using
  finite fields and dataflow graphs},
  \href{https://doi.org/10.1007/jhep07(2019)031}{\emph{Journal of High Energy
  Physics} {\bfseries 2019} (2019) 031}.

\bibitem{Bern:1993kr}
Z.~Bern, L.J.~Dixon and D.A.~Kosower, \emph{{Dimensionally regulated pentagon
  integrals}}, \href{https://doi.org/10.1016/0550-3213(94)90398-0}{\emph{Nucl.
  Phys. B} {\bfseries 412} (1994) 751}
  [\href{https://arxiv.org/abs/hep-ph/9306240}{{\ttfamily hep-ph/9306240}}].

\bibitem{Gluza:2010ws}
J.~Gluza, K.~Kajda and D.A.~Kosower, \emph{{Towards a Basis for Planar Two-Loop
  Integrals}}, \href{https://doi.org/10.1103/PhysRevD.83.045012}{\emph{Phys.
  Rev. D} {\bfseries 83} (2011) 045012}
  [\href{https://arxiv.org/abs/1009.0472}{{\ttfamily 1009.0472}}].

\bibitem{Schabinger:2011dz}
R.M.~Schabinger, \emph{{A New Algorithm For The Generation Of
  Unitarity-Compatible Integration By Parts Relations}},
  \href{https://doi.org/10.1007/JHEP01(2012)077}{\emph{JHEP} {\bfseries 01}
  (2012) 077} [\href{https://arxiv.org/abs/1111.4220}{{\ttfamily 1111.4220}}].

\bibitem{Lee:2014tja}
R.N.~Lee, \emph{{Modern techniques of multiloop calculations}},  in \emph{{49th
  Rencontres de Moriond on QCD and High Energy Interactions}}, pp.~297--300,
  2014 [\href{https://arxiv.org/abs/1405.5616}{{\ttfamily 1405.5616}}].

\bibitem{Bohm:2017qme}
J.~B\"ohm, A.~Georgoudis, K.J.~Larsen, M.~Schulze and Y.~Zhang, \emph{{Complete
  sets of logarithmic vector fields for integration-by-parts identities of
  Feynman integrals}},
  \href{https://doi.org/10.1103/PhysRevD.98.025023}{\emph{Phys. Rev. D}
  {\bfseries 98} (2018) 025023}
  [\href{https://arxiv.org/abs/1712.09737}{{\ttfamily 1712.09737}}].

\bibitem{Kosower:2018obg}
D.A.~Kosower, \emph{{Direct Solution of Integration-by-Parts Systems}},
  \href{https://doi.org/10.1103/PhysRevD.98.025008}{\emph{Phys. Rev. D}
  {\bfseries 98} (2018) 025008}
  [\href{https://arxiv.org/abs/1804.00131}{{\ttfamily 1804.00131}}].

\bibitem{Larsen:2015ped}
K.J.~Larsen and Y.~Zhang, \emph{{Integration-by-parts reductions from unitarity
  cuts and algebraic geometry}},
  \href{https://doi.org/10.1103/PhysRevD.93.041701}{\emph{Phys. Rev. D}
  {\bfseries 93} (2016) 041701}
  [\href{https://arxiv.org/abs/1511.01071}{{\ttfamily 1511.01071}}].

\bibitem{Bohm:2018bdy}
J.~B\"ohm, A.~Georgoudis, K.J.~Larsen, H.~Sch\"onemann and Y.~Zhang,
  \emph{{Complete integration-by-parts reductions of the non-planar hexagon-box
  via module intersections}},
  \href{https://doi.org/10.1007/JHEP09(2018)024}{\emph{JHEP} {\bfseries 09}
  (2018) 024} [\href{https://arxiv.org/abs/1805.01873}{{\ttfamily
  1805.01873}}].

\bibitem{Bendle:2019csk}
D.~Bendle, J.~B\"ohm, W.~Decker, A.~Georgoudis, F.-J.~Pfreundt, M.~Rahn et~al.,
  \emph{{Integration-by-parts reductions of Feynman integrals using Singular
  and GPI-Space}}, \href{https://doi.org/10.1007/JHEP02(2020)079}{\emph{JHEP}
  {\bfseries 02} (2020) 079}
  [\href{https://arxiv.org/abs/1908.04301}{{\ttfamily 1908.04301}}].

\bibitem{Mastrolia:2018uzb}
P.~Mastrolia and S.~Mizera, \emph{{Feynman Integrals and Intersection Theory}},
  \href{https://doi.org/10.1007/JHEP02(2019)139}{\emph{JHEP} {\bfseries 02}
  (2019) 139} [\href{https://arxiv.org/abs/1810.03818}{{\ttfamily
  1810.03818}}].

\bibitem{Frellesvig:2019uqt}
H.~Frellesvig, F.~Gasparotto, M.K.~Mandal, P.~Mastrolia, L.~Mattiazzi and
  S.~Mizera, \emph{{Vector Space of Feynman Integrals and Multivariate
  Intersection Numbers}},
  \href{https://doi.org/10.1103/PhysRevLett.123.201602}{\emph{Phys. Rev. Lett.}
  {\bfseries 123} (2019) 201602}
  [\href{https://arxiv.org/abs/1907.02000}{{\ttfamily 1907.02000}}].

\bibitem{Frellesvig:2019kgj}
H.~Frellesvig, F.~Gasparotto, S.~Laporta, M.K.~Mandal, P.~Mastrolia,
  L.~Mattiazzi et~al., \emph{{Decomposition of Feynman Integrals on the Maximal
  Cut by Intersection Numbers}},
  \href{https://doi.org/10.1007/JHEP05(2019)153}{\emph{JHEP} {\bfseries 05}
  (2019) 153} [\href{https://arxiv.org/abs/1901.11510}{{\ttfamily
  1901.11510}}].

\bibitem{Abreu:2019wzk}
S.~Abreu, R.~Britto, C.~Duhr, E.~Gardi and J.~Matthew, \emph{{From positive
  geometries to a coaction on hypergeometric functions}},
  \href{https://doi.org/10.1007/JHEP02(2020)122}{\emph{JHEP} {\bfseries 02}
  (2020) 122} [\href{https://arxiv.org/abs/1910.08358}{{\ttfamily
  1910.08358}}].

\bibitem{Frellesvig:2020qot}
H.~Frellesvig, F.~Gasparotto, S.~Laporta, M.K.~Mandal, P.~Mastrolia,
  L.~Mattiazzi et~al., \emph{{Decomposition of Feynman Integrals by
  Multivariate Intersection Numbers}},
  \href{https://doi.org/10.1007/JHEP03(2021)027}{\emph{JHEP} {\bfseries 03}
  (2021) 027} [\href{https://arxiv.org/abs/2008.04823}{{\ttfamily
  2008.04823}}].

\bibitem{Weinzierl:2020xyy}
S.~Weinzierl, \emph{{On the computation of intersection numbers for twisted
  cocycles}}, \href{https://doi.org/10.1063/5.0054292}{\emph{J. Math. Phys.}
  {\bfseries 62} (2021) 072301}
  [\href{https://arxiv.org/abs/2002.01930}{{\ttfamily 2002.01930}}].

\bibitem{Caron-Huot:2021iev}
S.~Caron-Huot and A.~Pokraka, \emph{{Duals of Feynman Integrals. Part II.
  Generalized unitarity}},
  \href{https://doi.org/10.1007/JHEP04(2022)078}{\emph{JHEP} {\bfseries 04}
  (2022) 078} [\href{https://arxiv.org/abs/2112.00055}{{\ttfamily
  2112.00055}}].

\bibitem{Chen:2022lzr}
J.~Chen, X.~Jiang, C.~Ma, X.~Xu and L.L.~Yang, \emph{{Baikov representations,
  intersection theory, and canonical Feynman integrals}},  2, 2022.

\bibitem{Chestnov:2022alh}
V.~Chestnov, F.~Gasparotto, M.K.~Mandal, P.~Mastrolia, S.J.~Matsubara-Heo,
  H.J.~Munch et~al., \emph{{Macaulay Matrix for Feynman Integrals: Linear
  Relations and Intersection Numbers}},  4, 2022.

\bibitem{Lee:2012cn}
R.N.~Lee, \emph{{Presenting LiteRed: a tool for the Loop InTEgrals REDuction}},
   2012.
\newblock 10.48550/ARXIV.1212.2685.

\bibitem{Tarasov:2004ks}
O.V.~Tarasov, \emph{{Computation of Grobner bases for two loop propagator type
  integrals}}, \href{https://doi.org/10.1016/j.nima.2004.07.104}{\emph{Nucl.
  Instrum. Meth. A} {\bfseries 534} (2004) 293}
  [\href{https://arxiv.org/abs/hep-ph/0403253}{{\ttfamily hep-ph/0403253}}].

\bibitem{Gerdt:2005qf}
V.P.~Gerdt and D.~Robertz, \emph{{A Maple package for computing Gr{\"o}bner
  bases for linear recurrence relations}},
  \href{https://doi.org/10.1016/j.nima.2005.11.171}{\emph{Nucl. Instrum. Meth.
  A} {\bfseries 559} (2006) 215}
  [\href{https://arxiv.org/abs/cs/0509070}{{\ttfamily cs/0509070}}].

\bibitem{Smirnov:2005ky}
A.V.~Smirnov and V.A.~Smirnov, \emph{{Applying Grobner bases to solve reduction
  problems for Feynman integrals}},
  \href{https://doi.org/10.1088/1126-6708/2006/01/001}{\emph{JHEP} {\bfseries
  01} (2006) 001} [\href{https://arxiv.org/abs/hep-lat/0509187}{{\ttfamily
  hep-lat/0509187}}].

\bibitem{Smirnov:2006tz}
A.V.~Smirnov, \emph{{An Algorithm to construct Grobner bases for solving
  integration by parts relations}},
  \href{https://doi.org/10.1088/1126-6708/2006/04/026}{\emph{JHEP} {\bfseries
  04} (2006) 026} [\href{https://arxiv.org/abs/hep-ph/0602078}{{\ttfamily
  hep-ph/0602078}}].

\bibitem{Smirnov:2006wh}
A.V.~Smirnov and V.A.~Smirnov, \emph{{S-bases as a tool to solve reduction
  problems for Feynman integrals}},
  \href{https://doi.org/10.1016/j.nuclphysbps.2006.09.032}{\emph{Nucl. Phys. B
  Proc. Suppl.} {\bfseries 160} (2006) 80}
  [\href{https://arxiv.org/abs/hep-ph/0606247}{{\ttfamily hep-ph/0606247}}].

\bibitem{Lee:2008tj}
R.N.~Lee, \emph{{Group structure of the integration-by-part identities and its
  application to the reduction of multiloop integrals}},
  \href{https://doi.org/10.1088/1126-6708/2008/07/031}{\emph{JHEP} {\bfseries
  07} (2008) 031} [\href{https://arxiv.org/abs/0804.3008}{{\ttfamily
  0804.3008}}].

\bibitem{Buch}
B.~Buchberger, \emph{An {A}lgorithm for {F}inding the {B}asis {E}lements of the
  {R}esidue {C}lass {R}ing of a {Z}ero {D}imensional {P}olynomial {I}deal},
  {\emph{J. Symbolic Comput.} {\bfseries 41} (2006) 475}.

\bibitem{plural}
V.~Levandovskyy and H.~Sch{\"o}nemann, \emph{P{LURAL}---a {C}omputer {A}lgebra
  {S}ystem for {N}oncommutative {P}olynomial {A}lgebras},  in \emph{Proceedings
  of the 2003 International Symposium on Symbolic and Algebraic Computation},
  pp.~176--183 (electronic), ACM, 2003.

\bibitem{Chyzak-1998-GBS}
F.~Chyzak, \emph{Gröbner bases, symbolic summation and symbolic integration},
  in \emph{Gröbner Bases and Applications (Linz, 1998)}, vol.~251 of
  \emph{London Mathematical Society Lecture Note Series}, (Cambridge),
  pp.~32--60, Cambridge Univ. Press (1998).

\bibitem{LoopIntegrals}
M.~Barakat, R.~Br{\"u}ser, T.~Huber and J.~Piclum, ``{LoopIntegrals}, compute
  master integrals using commutative and noncommutative methods from
  computational algebraic geometry.''
  \href{https://homalg-project.github.io/pkg/LoopIntegrals}{\texttt{https://homalg-project.github.io/}\discretionary{}{}{}\texttt{pkg/}\discretionary{}{}{}\texttt{LoopIntegrals}},
  Apr, 2022.

\bibitem{GAP4111}
The GAP~Group, \emph{{GAP -- Groups, Algorithms, and Programming, Version
  4.11.1}}, 2021.

\bibitem{homalg-project}
{homalg~project~authors}, ``The $\mathtt{homalg}$ project -- {A}lgorithmic
  {H}omological {A}lgebra.''
  (\url{https://homalg-project.github.io/prj/homalg_project}), 2003--2022.

\bibitem{singular431}
W.~Decker, G.-M.~Greuel, G.~Pfister and H.~Sch\"onemann, ``{\sc Singular}
  {4-3-1} --- {A} computer algebra system for polynomial computations.''
  \url{http://www.singular.uni-kl.de}, 2019.

\bibitem{Koutschan09}
C.~Koutschan, \emph{Advanced applications of the holonomic systems approach},
  Ph.D. thesis, Research Institute for Symbolic Computation (RISC), Johannes
  Kep\-ler University, Linz, Austria, 2009.

\bibitem{Koutschan10b}
C.~Koutschan, \emph{{HolonomicFunctions} (user's guide)},  Tech. Rep.
  \href{https://risc.jku.at/sw/holonomicfunctions/}{10-01}, RISC Report Series,
  Johannes Kep\-ler University, Linz, Austria (2010).

\bibitem{Hoff2015}
J.S.~Hoff, \emph{{Methods for multiloop calculations and Higgs boson production
  at the LHC}}, Ph.D. thesis, {Karlsruher Institut f{\"{u}}r Technologie
  (KIT)}, 2015.
\newblock 10.5445/IR/1000047447.

\bibitem{Pak_2011}
A.~Pak and A.~Smirnov, \emph{Geometric approach to asymptotic expansion of
  feynman integrals},
  \href{https://doi.org/10.1140/epjc/s10052-011-1626-1}{\emph{The European
  Physical Journal C} {\bfseries 71} (2011) 1626}.

\bibitem{Lee:2013mka}
R.N.~Lee, \emph{{LiteRed 1.4: a powerful tool for reduction of multiloop
  integrals}}, \href{https://doi.org/10.1088/1742-6596/523/1/012059}{\emph{J.
  Phys. Conf. Ser.} {\bfseries 523} (2014) 012059}
  [\href{https://arxiv.org/abs/1310.1145}{{\ttfamily 1310.1145}}].

\bibitem{Gerdt:2004kt}
V.P.~Gerdt, \emph{{Grobner bases in perturbative calculations}},
  \href{https://doi.org/10.1016/j.nuclphysbps.2004.09.011}{\emph{Nucl. Phys. B
  Proc. Suppl.} {\bfseries 135} (2004) 232}
  [\href{https://arxiv.org/abs/hep-ph/0501053}{{\ttfamily hep-ph/0501053}}].

\bibitem{1LoopBox}
M.~Barakat, R.~Brüser, T.~Huber and J.~Piclum, \emph{{The IBP relations of the
  one-loop box (\url{https://homalg-project.github.io/nb/1LoopBox/})}},  2022.

\bibitem{Ruijl:2017dtg}
B.~Ruijl, T.~Ueda and J.~Vermaseren, \emph{Form version 4.2},  2017.
\newblock 10.48550/ARXIV.1707.06453.

\bibitem{Barakat:2022ttc}
M.~Barakat, R.~Br\"user, T.~Huber and J.~Piclum, \emph{{IBP reduction via
  Gr\"obner bases in a rational double-shift algebra}},
  \href{https://doi.org/10.22323/1.416.0043}{\emph{PoS} {\bfseries LL2022}
  (2022) 043} [\href{https://arxiv.org/abs/2207.09275}{{\ttfamily
  2207.09275}}].

\end{thebibliography}

\newpage

\end{document}